\documentclass[11pt]{article}
\usepackage[margin=1in]{geometry} % Gives standard 1-inch margins

\usepackage[utf8]{inputenc}
\usepackage[round]{natbib}
\usepackage{amsmath,amsfonts,amssymb,amscd,amstext,amsthm}
\usepackage{algorithmicx, algorithm,algpseudocode}
\usepackage{comment}
\usepackage{xspace}
\usepackage{multirow}
\usepackage{mathtools}
\usepackage{mathrsfs}
\usepackage{latexsym}
\usepackage[usenames,dvipsnames]{xcolor}
\usepackage{pdfpages}
\usepackage{enumerate}
\usepackage{url}
\usepackage[hyperfootnotes=false]{hyperref}
\usepackage{authblk} % Added to handle multiple affiliations cleanly

\usepackage{tikz}
\usetikzlibrary{angles,quotes,shapes,calc,positioning,intersections,math,through,backgrounds,arrows.meta}

\newtheorem{theorem}{Theorem}
\newtheorem{lemma}{Lemma}
\newtheorem{corollary}{Corollary}

\newtheorem{definition}{Definition}

\long\def\IGNORE#1{}

\newcommand{\OPT}{\textsc{Opt}\xspace}

\newcommand{\SU}{\textsc{Straight-Up}\xspace}
\newcommand{\Greedy}{\textsc{Greedy}\xspace}
\newcommand{\ConeApex}{\textsc{Cone-Apex}\xspace}

\definecolor{DarkGreen}{rgb}{0.0, 0.5, 0.0}

\newcommand{\norm}[1]{\left\lVert#1\right\rVert}

\newcommand{\ignore}[1]{}

\title{Online Drone Coverage of Targets on a Line}

\author[1]{Stefan Dobrev\thanks{S. Dobrev was supported in part by VEGA 2/0117/25, and K. Georgiou, E. Kranakis, L. Narayanan, D. Pankratov were supported in part by NSERC.}}
\author[2]{Konstantinos Georgiou}
\author[3]{Evangelos Kranakis}
\author[4]{Danny Krizanc}
\author[5]{Lata Narayanan}
\author[5]{Jaroslav Opatrny}
\author[5]{Denis Pankratov}
\author[6]{Sunil Shende\thanks{This is the full version of the paper that appeared in the Proceedings of the 37th International Workshop on Combinatorial Algorithms (IWOCA 2026), Clermont-Ferrand, France, on June 8-12, 2026. It corrects and extends several results from the initial ArXiv version of the full conference paper.}}

\affil[1]{Slovak Academy of Sciences, Bratislava, Slovakia}
\affil[2]{Dept. of Math., Toronto Metropolitan University, Toronto, Canada}
\affil[3]{School of Computer Science, Carleton University, Ottawa, Canada}
\affil[4]{Dept. of Math. and Computer Science, Wesleyan University, Middletown CT, USA}
\affil[5]{Dept. of Comp. Sc. and Soft. Eng., Concordia University, Montreal, QC, Canada}
\affil[6]{Department of Computer Science, Rutgers University, Camden, USA}

\date{} % Leaves the date blank. Remove this line if you want today's date to appear.

\begin{document}

\maketitle

\begin{abstract}
We study the online target coverage problem for a mobile drone or sensor equipped with a camera or directional antenna of fixed scanning half-angle $\alpha \in (0, \pi/2)$. Targets appear sequentially at arbitrary positions along a one-dimensional line barrier, and upon each arrival the drone must move to a feasible position whose visibility interval covers both the newly revealed target and all previously arrived targets. The objective is to minimize the total Euclidean length of the drone's trajectory, evaluated against an optimal offline trajectory via the competitive ratio.

To analyze both offline and online trajectories uniformly, we reformulate the geometry in \emph{coverage coordinates}, mapping slanted visibility cones to axis-aligned quadrants equipped with an induced Euclidean metric norm. Using this framework, we establish exact competitive ratios across the full angle domain $\alpha \in (0, \pi/2)$ for three natural baseline algorithms,\textsc{Straight-Up}, \textsc{Cone-Apex}, and \textsc{Greedy}, all of which attain a competitive ratio of $\sqrt{2}$ at $\alpha = \pi/4$. We then analyze fixed-angle steering strategies and derive the optimal fixed-angle algorithm $\beta^*$-\textsc{Hedge}, which improves upon the baselines for $\alpha \in (\pi/6, \pi/3)$ and peaks at $5/4$ when $\alpha = \pi/4$. Going further, we introduce a family of memoryless, state-based linear interpolation algorithms $\mathcal{A}_\lambda$ (generalizing the \textsc{Midpoint} rule at $\alpha = \pi/4$) and prove that the optimized algorithm $\mathcal{A}_{\lambda^*}$ strictly outperforms $\beta^*$-\textsc{Hedge} throughout $\alpha \in (0, \pi/3)$, achieving a tight competitive ratio of $\sqrt{3/2} \approx 1.2247$ at $\alpha = \pi/4$ while matching the optimal bound $\csc\alpha$ for $\beta^*$-\textsc{Hedge} on $[\pi/3, \pi/2)$.

Finally, we establish a universal lower bound on the competitive ratio of any deterministic or randomized online algorithm (against an oblivious adversary) across the entire domain $\alpha \in (0, \pi/2)$. This lower bound attains its maximum at $\alpha = \pi/4$ with value $(1 + \sqrt{2})/2 \approx 1.2071$, narrowing the optimal competitive ratio at $\alpha = \pi/4$ to the interval $[(1+\sqrt{2})/2,\, \sqrt{3/2}]$.

\vspace{0.5cm}
\end{abstract}

\section{Introduction} \label{sec:intro}

We consider the problem of \emph{path planning for online target coverage} by an
autonomous drone: given a sequence of target points that appear \emph{online}
on a line on the ground, determine a competitive trajectory for the drone so that
it covers the collection of target requests received so far. Each successive point,
called a \emph{request}, must be \textit{serviced} by moving the drone into a
position that covers the new request as well as all previous requests.

Displacing the drone from one position to another consumes energy; we assume that
the energy consumed is proportional to the total Euclidean length of its trajectory,
which is standard in the literature~\cite{shivgan2020energy,zorbas2016optimal}.
The goal of our online algorithms is to minimize the total sum of displacements
needed to cover all requests.

\subsection{Related work}

Area, barrier, and target coverage are key problem domains in wireless sensor
networks and have been extensively studied for isotropic sensors with circular
sensing ranges~\cite{Survey2014,Coverage2001,Survey2011,Wang2009}. To model
directional cameras, Shih et al.~\cite{Shih2010} introduced barrier coverage
using stationary sensors whose sensing range is a sector of angle $\theta$ (field
of view) with fixed orientation, while Ma et al.~\cite{Ma2012} allowed stationary
cameras to be oriented so as to minimize the number of sensors forming a barrier.
Similar stationary directional models have been investigated for target coverage:
Johnson and Bar-Noy~\cite{johnson2011pan} studied the \emph{pan-and-scan} problem
of choosing camera orientations and focal lengths in the plane to maximize sensing
quality (proving NP-hardness in general and giving approximation and polynomial-time
special cases); Yao et al.~\cite{YGLP2021} studied minimizing the number of
rotatable stationary cameras needed to cover point targets on a line; Munishwar
et al.~\cite{Munishwar2010} studied discrete pan selection for stationary cameras
to maximize covered planar targets; and Neishaboori et al.~\cite{Neishaboori2014}
gave heuristics to minimize the number of fixed-angle, fixed-range cameras covering
planar points. Unlike our setting, all of these camera models rely on multiple
stationary sensors rather than a single mobile agent. Conversely,
movement-minimization algorithms for mobile target coverage have focused on
isotropic disk sensors~\cite{Chen2016,Liao2015,Liao2012} rather than directional
cameras.

Target and area coverage with aerial drones has also been studied primarily in
offline or static settings. Saeed et al.~\cite{Saeed2019} showed that covering
oriented planar line segments using a minimum number of drones equipped with
orientable cameras is NP-hard and hard to approximate within $O(\log n)$. Zorbas
et al.~\cite{zorbas2016optimal} introduced the minimum-cost drone location problem
in a two-dimensional terrain for drones with circular sensing ranges and
energy-dependent altitude, while Savkin et al.~\cite{savkin2019method} developed a
distributed optimization algorithm to position drones for maximum coverage quality
subject to maintaining a connected communication graph with ground nodes. Other
work addresses static trajectory planning for points of interest~\cite{Bezas2022}
or fleets of drones continuously monitoring mobile ground sensors and relaying data
to a base station~\cite{caillouet2018optimization}.

From an online optimization perspective, the closest theoretical framework is
\emph{Convex Body Chasing}~\cite{Friedman,Bubeck}, in which a guarding point moves
online with minimal total Euclidean movement so as to remain within a sequence of
revealed convex bodies. Although our feasibility cones form a nested sequence of
two-dimensional convex sets, general convex body chasing bounds do not yield tight
competitive ratios for our specific cone geometry. To our knowledge, this work is
the first to study \emph{online} target coverage by a mobile drone or sensor with
a fixed-angle directional field of view.
\section{Preliminaries}

In our model, a drone can move and hover aerially at any height in the upper half-plane $\mathcal{H} = \{(x, y) \in \mathbb{R}^2 : y \ge 0\}$, where $x \in \mathbb{R}$ denotes horizontal position along a ground line and $y \ge 0$
represents altitude. We refer to $\mathcal{H}$ as the \textbf{(physical) position space}: distances in this space are measured using the standard Euclidean $L_2$ distance norm $\norm{\cdot}$. The drone is equipped with a camera or scanning sensor that points straight at the ground with a \textit{fixed field-of-view} defined by the \textbf{scanning half-angle} $\alpha \in \left(0, \frac{\pi}{2}\right)$ with respect to the vertical axis. We refer to the angle $2\alpha$ as the scanning angle of the drone.

When positioned at $P = (x, y) \in \mathcal{H}$, the drone's scanning cone intersects the $x$-axis to cover a ground segment of length $2y\tan\alpha$: this is the drone's \textbf{coverage footprint}. The actual \textbf{covered interval} $I(x, y)$ on the $x$-axis is given by:

\begin{equation*}
I(x, y) = [x - y \tan\alpha, \; x + y \tan\alpha]
\end{equation*}
Note that $\alpha \ge \frac{\pi}{2}$ would imply that at any height $y > 0$, the drone trivially covers the $x$-axis, so $\alpha$ is strictly between 0 and $\pi/2$.

In our problem, an adversary reveals the finite request sequence
$X = (x_0=0, x_1, x_2 \ldots,x_k)$ with $x_i$ revealed in time step $i$: at time $i$, we define the following:
\begin{itemize}
\item Left extreme bound $a_i = - \; \min_{0 \le j \le i} x_j$
\item Right extreme bound $b_i =  \; \max_{0 \le j \le i} x_j$
\item Ground Interval $G_i = [-a_i, \; b_i]$
\item Dominant extreme bound $m_i = \max (a_i, b_i)$
\item Non-dominant extreme bound $n_i = \min (a_i, b_i)$
\end{itemize}
We assume that each new request is non-redundant: $x_{i+1}$ is always outside the ground interval $G_i$. Suppose that the sequence ends after $k$ steps. At this stage, we denote by $m = \max (a_k, b_k)$ and $n=\min (a_k, b_k)$ the \textbf{dominant} and \textbf{non-dominant} extremes of the entire sequence, and define the parameter $r = n/m$, the asymmetry of the sequence.

The drone starts at position $P_0 = (0,0)$ at time step 0. It can expand and translate its coverage footprint by increasing its altitude $y$ or moving horizontally parallel to the ground. To satisfy coverage, the drone's position, $P_i$, just before step $i+1$, must satisfy the invariant that its scanning cone covers $G_i$.
It is possible that the drone at its current position $P_{i}$ already covers $x_{i+1}$ even if it is non-redundant. In that case, the drone does not need to move at all, i.e., $P_{i+1} = P_i$ is feasible. Otherwise, $x_{i+1}$ elicits either an \textit{active left move} if $x_{i+1} < -a_{i}$ or an \textit{active right move} if $x_{i+1} > b_i$: the online algorithm must move to expand its corresponding coverage footprint to the left i.e., $a_{i+1} = -x_{i+1}$ after an active left move, or $b_{i+1} = x_{i+1}$ after an active right move.

The set of all valid drone positions in $\mathcal{H}$ that cover $G_i$ forms the upward-facing \textbf{feasibility cone} $K_\alpha(a_i, b_i) \subset \mathcal{H}$ given by:
\begin{align*}
    K_\alpha(a_i, b_i) &= \{(x, y) \in \mathcal{H} : x - y\tan\alpha \le -a_i \text{ and } x + y\tan\alpha \ge b_i \}
\end{align*}
with cone apex
\begin{align*}
T_i = \left( \frac{b_i - a_i}{2}, \; \frac{a_i + b_i}{2} \cot\alpha \right)
\end{align*}
It is easy to verify that $K_\alpha(a_i, b_i)$ is convex with $T_i$ being the unique point minimizing altitude $y$ in $K_\alpha(a_i, b_i)$. The cone has two flanks: one flank has positive slope $\tan (\pi/2 - \alpha) = \cot \alpha$ that is incident on $T_i$ and the point $-a_i$ on the $x$-axis, while the other flank has negative slope $\tan (\pi/2 + \alpha) = - \cot \alpha$ and is incident on $T_i$ and the point $b_i$ on the $x$-axis.

\begin{figure}[htbp]
\centering
\begin{tikzpicture}[
    scale=1.1,
    >=Stealth,
    point/.style={circle, fill=black, inner sep=1.2pt},
    apex/.style={circle, fill=blue!80!black, inner sep=1.5pt},
    optline/.style={thick, dashed, red!80!black},
    algline/.style={thick, blue!80!black, ->},
    coneline/.style={gray!50, dashed, thin}
]
    % Axes
    \draw[->, thick] (-3.0, 0) -- (5.0, 0) node[right] {$x$};
    \draw[->, thick] (0, 0) -- (0, 4.0) node[above] {$y$};

    % Target Requests on x-axis
    \node[point, label=below:{\footnotesize $x_0 (0)$}] at (0,0) {};
    \node[point, label=below:{\footnotesize $x_1 (1)$}] at (1,0) {};
    \node[point, label=below:{\footnotesize $x_2 (-2)$}] at (-2,0) {};
    \node[point, label=below:{\footnotesize $x_3 (4)$}] at (4,0) {};

    % Feasibility Cones Boundaries (alpha = 45 deg, slope = +/- 1)
    \draw[coneline] (-2,0) -- (2,4);
    \draw[coneline] (1,0) -- (-3,4);
    \draw[coneline] (4,0) -- (0,4);
    \draw[coneline] (0,0) -- (4,4);

    % Apex Points
    \node[apex, label=right:{\footnotesize $T_1$}] (T1) at (0.5, 0.5) {};
    \node[apex, label=left:{\footnotesize $T_2$}] (T2) at (-0.5, 1.5) {};
    \node[apex, label=above:{\footnotesize $T_3$}] (T3) at (1.0, 3.0) {};

    % --- Half-angle alpha at T3 ---
    \draw[gray, dotted, thin] (1.0, 3.0) -- (1.0, 1.8);
    \draw[thin] (1.0, 2.4) arc[start angle=270, end angle=315, radius=0.6]
      node[midway, below=2pt, font=\tiny] {$\alpha$};
\end{tikzpicture}
\caption{Feasibility cone apexes for sequence $X=(0, 1, -2, 4)$ for $\alpha = \pi/4$.}
\label{fig:fcones}
\end{figure}

Figure~\ref{fig:fcones} illustrates the progression of nested feasibility cones on a geometric request sequence $X=(0, 1, -2, 4)$ for a drone with half-angle $\alpha = \pi/4$ (i.e., with $\tan \alpha = 1$). For example, as shown in the figure, a drone positioned at the apex $T_2$ will cover the interval $[x_2, x_1] = [-2, 1]$, the extreme requests in the request sequence prefix $(x_0, x_1, x_2)$.

\subsection{Competitive Analysis Framework}

To maintain coverage, our online algorithms make an irrevocable decision during time step $i$ to move from position
$P_{i-1} \in K_{\alpha}(a_{i-1}, b_{i-1})$
to a position
$P_{i} \in K_{\alpha}(a_i, b_i)$
incurring a step cost of
$\norm{P_i - P_{i-1}}$,
the $L_2$ Euclidean distance between the positions.

We seek to minimize the cumulative cost incurred by such an online algorithm $\text{ALG}$ moving through positions $P_0, P_1, \ldots, P_k$ where $P_i$ is a position on the feasibility cone $K_{\alpha}(a_i, b_i)$. This cost is given by

\begin{align}
\text{ALG}(X) &= \sum_{i=1}^k \norm{P_i - P_{i-1}}
\end{align}
which is its total path length along the path $P_0, P_1, \ldots, P_k$.
If all the points in the request sequence are known in advance, the optimal (offline) algorithm \OPT starts at $P_0$, the origin, and goes directly to the closest point on the feasibility cone $K_{\alpha} (a_k, b_k)$.

Let $\mathcal{X}^*$ be the set of all finite sequences. We say that $\mathrm{ALG}$ is \emph{$\rho$-competitive} if
\[
\mathrm{ALG}(X) \le \rho \cdot \mathrm{OPT}(X)
\qquad\text{for all } X \in \mathcal{X}^*,
\]
and its \emph{competitive ratio} is
\[
\rho(\mathrm{ALG}; \alpha) := \sup_{X \in \mathcal{X}^*} \frac{\mathrm{ALG}(X)}{\mathrm{OPT}(X)}.
\]
The \emph{optimal deterministic competitive ratio} at angle $\alpha$ is
\[
\rho^*_{\mathrm{det}}(\alpha) := \inf_{\mathrm{ALG}} \rho(\mathrm{ALG}; \alpha),
\]
where the infimum is taken over all feasible deterministic online algorithms.

\subsection{Correspondence Between Coverage Space And Position Space}

While the drone moves in position space $\mathcal{H}$, its drone coverage footprints as well as request intervals are over the ground (i.e., the $x$-axis). Since a drone at position $(x, y)$ must cover the origin $x_0=0$, the following one-to-one correspondence can be established between its position $(x,y)$ and its corresponding \textit{absolute coverage extents} $(u, v)$ to the left and right of the origin, i.e., $I(x,y) = [-u, v]$ with $u, v \ge 0$:
\begin{align}
(u, v) &= (y \tan \alpha - x, \; y \tan \alpha + x) \nonumber \\
(x, y) &= \left( \frac{v - u}{2}, \; \frac{u+v}{2} \cot \alpha \right) \label{eq:uv_xy}
\end{align}
Thus, every point $(x, y)$ in position space $\mathcal{H}$ corresponds bijectively to point $(u, v)$ in \textbf{coverage space} ($u, v \ge 0$) via this linear transformation.

\begin{definition}
Consider a drone with half-scanning angle $\alpha \in (0, \pi/2)$ that has coordinates $(u,v)$ in coverage space. With $c = \cos (2 \alpha)$, let
\begin{align}
N_{\alpha}(u, v) &= \frac{\sqrt{u^2 + v^2 + 2c uv}}{2 \sin \alpha} \label{eq:N_alpha}
\end{align}
\end{definition}
In fact, $N_{\alpha}$ defines a \textit{positive-definite distance norm} on coverage coordinates:
\begin{align*}
N_{\alpha}(u, v) &\ge 0 \quad \text{with } N_{\alpha}(u, v) = 0 \iff (u, v) = (0, 0), \quad \text{(Positive definiteness)}\\
N_{\alpha}(u, v) &= N_{\alpha}(v, u) \quad \text{(Isometry under coordinate swap)}, \\
N_{\alpha}(k u, k v) &= |k| N_{\alpha}(u, v) \quad \text{for all } k \in \mathbb{R} \quad \text{(Absolute scalability)}.
\end{align*}
Furthermore, $N_{\alpha}$ satisfies sub-additivity for any step increments in coverage space:
$$
N_{\alpha}(u_1 + u_2, \, v_1 + v_2) \le N_{\alpha}(u_1, v_1) + N_{\alpha}(u_2, v_2)
$$
Indeed, a quick calculation shows that $L_2$ distance $\norm{(x,y)}=\sqrt{x^2 + y^2}$ can be rewritten as:
\begin{alignat*}{3}
&\norm{(x,y)} &&= \frac{1}{2} \sqrt{(v-u)^2 + (u+v)^2 {\cot}^2 \alpha} &&= \frac{1}{2 \sin \alpha} \sqrt{u^2 + v^2 + 2uv({\cos}^2 \alpha - {\sin}^2 \alpha)}\\
& &&= \frac{\sqrt{u^2 + v^2 + 2c uv}}{2 \sin \alpha} &&= N_{\alpha}(u, v)
\end{alignat*}

As will become clear, the correspondence between Euclidean norm and $N_{\alpha}(\cdot)$ is extremely useful: it permits us  to measure movement using coverage space coordinates rather than using position space coordinates.

\begin{figure}[htbp]
\centering
\begin{tikzpicture}[
    scale=1.1,
    >=Stealth,
    point/.style={circle, fill=black, inner sep=1.2pt},
    apex/.style={circle, fill=blue!80!black, inner sep=1.5pt},
    optline/.style={thick, dashed, red!80!black},
    algline/.style={thick, blue!80!black, ->},
    coneline/.style={gray!50, dashed, thin}
]

  % =========================================================
  % SUBFIGURE 1: (x,y) Physical Space
  % =========================================================
  \begin{scope}[shift={(0,0)}]

    % Axes
    \draw[->, thick] (-3.0, 0) -- (5.0, 0) node[right] {$x$};
    \draw[->, thick] (0, 0) -- (0, 4.0) node[above] {$y$};

    % Target Requests on x-axis
    \node[point, label=below:{\footnotesize $x_0$}] at (0,0) {};
    \node[point, label=below:{\footnotesize $x_1$}] at (1,0) {};
    \node[point, label=below:{\footnotesize $x_2$}] at (-2,0) {};
    \node[point, label=below:{\footnotesize $x_3$}] at (4,0) {};

    % Feasibility Cones Boundaries (alpha = 45 deg, slope = +/- 1)
    \draw[coneline] (-2,0) -- (2,4);
    \draw[coneline] (1,0) -- (-3,4);
    \draw[coneline] (4,0) -- (0,4);
    \draw[coneline] (0,0) -- (4,4);

    % Apex Points
    \node[apex, label=right:{\footnotesize $T_1(0.5, 0.5)$}] (T1) at (0.5, 0.5) {};
    \node[apex, label=left:{\footnotesize $T_2(-0.5,1.5)$}] (T2) at (-0.5, 1.5) {};
    \node[apex, label=above:{\footnotesize $T_3(1,3)$}] (T3) at (1.0, 3.0) {};

    % OPT Trajectory (Straight to final apex T3)
    \draw[optline] (0,0) -- (T3) node[midway, right, font=\footnotesize] {\OPT};

    % Greedy Trajectory (P0 -> T1 -> T2 -> T3)
    \draw[algline] (0,0) -- (T1);
    \draw[algline] (T1) -- (T2);
    \draw[algline] (T2) -- (T3) node[midway, left, font=\footnotesize] {\ConeApex};
  \end{scope}

  % =========================================================
  % SUBFIGURE 2: (u,v) Coverage Space
  % =========================================================
  \begin{scope}[shift={(7.5,0)}]

    % Axes
    \draw[->, thick] (0, 0) -- (3.5, 0) node[right] {$u$};
    \draw[->, thick] (0, 0) -- (0, 4.8) node[above] {$v$};

    % Target Feasibility Bounds (u >= 2, v >= 4)
    \draw[gray!40, dotted, thick] (2,0) -- (2,4.5) node[above, black, font=\tiny] {$u=2$};
    \draw[gray!40, dotted, thick] (0,4) -- (3.2,4) node[right, black, font=\tiny] {$v=4$};
    \draw[gray!40, dotted, thick] (2,1) -- (3.2,1) node[right, black, font=\tiny] {$v=1$};

    % Apex Points in (u,v) Space
    \node[point, label=below left:{\footnotesize $(0,0)$}] (U0) at (0,0) {};
    \node[apex, label=left:{\footnotesize $T_1^{'}(0,1)$}] (U1) at (0,1) {};
    \node[apex, label=below right:{\footnotesize $T_2^{'}(2,1)$}] (U2) at (2,1) {};
    \node[apex, label=above right:{\footnotesize $T_3^{'}(2,4)$}] (U3) at (2,4) {};

    % OPT Trajectory (Straight diagonal line)
    \draw[optline] (0,0) -- (U3) node[midway, left, font=\footnotesize] {\OPT};

    % Greedy Trajectory (Axis-aligned staircase)
    \draw[algline] (0,0) -- (U1);
    \draw[algline] (U1) -- (U2);
    \draw[algline] (U2) -- (U3) node[midway, right, font=\footnotesize] {\ConeApex};
  \end{scope}

\end{tikzpicture}
\caption{Trajectories of \ConeApex and \OPT on sequence $(0, 1, -2, 4)$ for $\alpha = \pi/4$ (a) In Position Space $(x,y)$ (b) In Coverage Space $(u,v)$
with norm $N_{\pi/4}(\cdot)$}
\label{fig:cone_apex_trajectory}
\end{figure}

For instance, Figure~\ref{fig:cone_apex_trajectory} depicts the trajectories of two algorithms on the sequence $X=(0, 1, -2, 4)$ with $\alpha=\pi/4$:
the \textbf{online} \ConeApex strategy\footnote{As the name suggests, \ConeApex simply moves from one cone apex to the next. We will describe and analyze it in more detail in the next section.} that starts at $P_0=T_0$ and follows the feasible path $P_0,~P_1=T_1,~P_2=T_2,~P_3=T_3$ in position space. The \textbf{offline} \OPT strategy, knowing the sequence in advance, goes straight to $T_3$. Figure~\ref{fig:cone_apex_trajectory}(b) shows the corresponding points and trajectories in coverage space, e.g., $T_3^{'} = (2,4)$ is the image of cone apex $T_3$ in coverage space. We note that \ConeApex follows a rectilinear path in coverage space: the three segments $T_i^{'} T_{i+1}^{'}$ for $i \in \{0,1,2\}$, have norms of the form $N_{\pi/4}(\Delta u, 0)=(\Delta u)/\sqrt{2}$ or $N_{\pi/4}(0, \Delta v)=(\Delta v)/\sqrt{2}$ since $\cos (2 \alpha) = 0$ and $2 \sin \alpha = \sqrt{2}$. Hence, it is easy to check that
\begin{alignat*}{3}
& \text{\ConeApex}(X) &&= \sum_{i=1}^3 \norm{T_{i} - T_{i-1}}
&&= \sum_{i=1}^3 N_{\pi/4}(T_{i}^{'} - T_{i-1}^{'}) \\
& &&= \frac{1}{\sqrt{2}} (1 + 2 + 3) &&= 3\sqrt{2}
\end{alignat*}
while $OPT(X) = N_{\pi/4}(T_3^{'}) = \sqrt{10}$

\subsection{A Characterization Of \textsc{OPT}}

The optimal offline algorithm, $\OPT$, knows all the points in the request sequence $X=(x_0=0, x_1, \ldots, x_k)$ in advance. Hence, it moves directly from the origin to the point on the boundary of the feasibility cone $K_{\alpha}(a_k, b_k)$ closest to the origin. Working in coverage space simplifies computing $\OPT(X)$ because the cone boundary of $K_{\alpha}(a_k, b_k)$ in position space maps precisely to the two axis-parallel boundary rays $u = a_k, v \ge b_k$ and $u \ge a_k, v = b_k$ in coverage space.

\begin{theorem}
Consider the request sequence $X=(x_0=0, x_1, \ldots, x_k)$ with extreme points $-a_k$ and $b_k$. Let $m = \max (a_k, b_k) > 0$, $n = \min(a_k, b_k) \ge 0$, $r = n/m \in [0, 1]$ and $c = \cos (2 \alpha) \in (-1, 1)$.
\begin{enumerate}[(a)]
  \item If $r \geq \max(0, -c)$, \OPT moves to the cone apex of $K_{\alpha}(a_k, b_k)$ with
  \[ \OPT(X) = \frac{m}{2 \sin \alpha} \sqrt{1+r^2 + 2cr} \]
  \item Otherwise, if $c < 0$ and $r < -c$, \OPT goes to the point on the nearest flank of the cone $K_{\alpha}(a_k, b_k)$ in a trajectory perpendicular to the flank. In this case, its cost is
  \[ \OPT(X) = m \cos \alpha \]
\end{enumerate}
\label{thm:opt}
\end{theorem}

\begin{proof}
Let $P_{\OPT}=(x,y)$ be the point on the boundary of the feasibility cone $K_{\alpha} (a_k, b_k)$ that minimizes $\norm{(x,y)}$. Let $(u,v)$ be the corresponding point in coverage space from the bijection given in~\eqref{eq:uv_xy}. Since $\norm{(x,y)} = N_{\alpha}(u,v)$, the point $(u,v)$ is a feasible point in coverage space that minimizes $N_{\alpha}(u,v)$.

Without loss of generality, we may assume that $a_k \ge b_k$, so $a_k=m$ and $b_k=n$. From~\eqref{eq:N_alpha}, minimizing $N_{\alpha}(u,v)$ is equivalent to minimizing $4 {\sin}^2 \alpha N_{\alpha}^2(u,v) := f(u,v) = u^2 + v^2 + 2cuv$ under the constraints $u \ge m$ and $v \ge n$. Consider the partial derivatives of $f(u,v)$:
$$\left(\frac{\partial f}{\partial u}, \frac{\partial f}{\partial v}\right) = (2u + 2cv, \;2v + 2cu)$$
When we evaluate the partial derivatives at the cone apex position $(m, n)$ in coverage space, we get:
$$(\frac{\partial f}{\partial u}, \frac{\partial f}{\partial v})\Big\vert{}_{(m,n)} = (2m + 2cn, \; 2n + 2cm) = 2m (1 + cr, r+c)$$
Regardless of the angle $\alpha$, we have $1 + cr > 0$ because $c = \cos (2 \alpha) \in (-1, 1)$ and $r \in [0,1]$. Hence, the sign of the $r+c$ determines whether moving along the cone flank $u=m$ away from the cone apex $(m, n)$ in coverage space will increase $N_{\alpha}(u,v)$.

\begin{enumerate}[(a)]
\item If $c \ge 0$ or $r \ge -c$ then $r + c \ge 0$, whence $f(u, v)$ is minimized at $(m, n)$. The cone apex is the closest point to the origin and \OPT will move there with cost
$$\OPT(X) = N_{\alpha}(m, n) = \frac{m}{2 \sin \alpha} \sqrt{1 + r^2 + 2cr}$$

\item Otherwise, $c < 0$ and $0 \leq r < -c$. Since $r + c < 0$, the partial derivative of $f(u,v)$ with respect to $v$ is negative at $(m, n)$ in coverage space. So $f$ strictly decreases along $u = m$ as $v$ increases past $n$.  Holding $u = m$ fixed, we minimize $g(v) := f(m,v) = m^2 + v^2 + 2cm v$. Setting $g'(v) = 2v + 2cm = 0$ and verifying that $g''(v)=2 > 0$ means that $g(v)$ is minimized at $v = -cm = |c|m > rm = n$. Since $f(m, \: |c|m) = m^2 + |c|^2 m^2 + 2c|c|m^2 = m^2(1 - c^2)$, we get
\begin{alignat*}{3}
& \OPT(X) &&= N_{\alpha} (m, \: |c|m) &&= \frac{1}{2 \sin \alpha} \sqrt{f(m,\: |c|m)} \\
& &&= \frac{m}{2 \sin \alpha} \sqrt{1 - c^2} &&= \frac{m}{2 \sin \alpha} \sin (2 \alpha) \\
& &&= m \cos \alpha &&
\end{alignat*}
Geometrically, when $c < 0$ and $0 \leq r < -c$, the point closest to the final cone from the origin is not the cone apex, but is the point reached perpendicularly from the origin to the cone's nearest flank. We can verify this by translating the point's coordinates $(m, |c|m)$ in coverage space to its coordinates in position space using~\eqref{eq:uv_xy}:
\[ (x, y) = \left(\frac{|c|m - m}{2}, \; \frac{|c|m + m}{2} \cot \alpha\right) \]
A quick calculation shows that $y/x$, the slope of \OPT's trajectory, simplifies to $-\tan \alpha$. Since the cone flank corresponding to $u=m$ has slope $\cot \alpha$, \OPT must reach it at a right angle.
\end{enumerate}
\end{proof}
We remark that \OPT's transition from the \textit{apex regime} in case (a) above to the \textit{off-apex regime} in case (b) is continuous: at the transition threshold $r = -c$, the two branches of $\mathrm{OPT}(X)$ coincide at $m\cos\alpha$.

\subsection{Outline and Results}

Armed with the characterization of $\textsc{Opt}$ in Theorem~2, we analyze three natural memoryless baseline strategies in Section~3, \textsc{Straight-Up}, \textsc{Cone-Apex}, and \textsc{Greedy}, and show that all three belong to a unified family $\mathcal{A}_\beta$ of fixed-angle steering algorithms parameterized by a steering angle $\beta \in [0,\alpha]$ from the vertical. Using a single over-coverage accounting lemma, we establish their exact competitive ratios for every $\alpha \in (0,\pi/2)$: $2\cos\alpha$ (for $\alpha \le \pi/4$) and $\csc\alpha$ (for $\alpha > \pi/4$) for \textsc{Straight-Up}; $\sec\alpha$ for \textsc{Cone-Apex}; and $\sec\alpha$ (for $\alpha \le \pi/4$) and $\sqrt{2 - \cos^2(2\alpha)}$ (for $\alpha > \pi/4$) for \textsc{Greedy}.

Because \textsc{Cone-Apex} ($\beta=\alpha$) excels on one-sided request sequences ($r \approx 0$) while \textsc{Straight-Up} ($\beta=0$) excels on symmetric sequences ($r \approx 1$), Section~4 studies the full fixed-angle family $\beta$-\textsc{Hedge} ($\beta \in [0,\alpha]$), which interpolates between these two extremes. We derive the exact competitive ratio of $\beta$-\textsc{Hedge} for every pair $(\alpha,\beta)$ and determine the unique steering angle $\beta^*(\alpha)$ that minimizes it. The resulting optimal fixed-angle strategy, $\beta^*$-\textsc{Hedge}, coincides with \textsc{Cone-Apex} on $(0,\pi/6]$, coincides with \textsc{Straight-Up} on $[\pi/3,\pi/2)$, and achieves competitive ratio $\sin(2\alpha) + \frac{1}{4\sin(2\alpha)}$ on $(\pi/6,\pi/3)$, peaking at $5/4 = 1.25$ when $\alpha = \pi/4$ (improving upon the $\sqrt{2}$ ratio of all three baselines).

Fixed-angle steering is inherently path-dependent and retains excess over-coverage on the dominant side. In Section~5, we eliminate this inefficiency via a memoryless state-based family $\mathcal{A}_\lambda$ ($\lambda \in [0,1]$) that moves directly to a convex combination of the current cone apex and the vertical \textsc{Straight-Up} position. Optimizing $\lambda$ yields a competitive-ratio envelope that matches $\beta^*$-\textsc{Hedge} on $[\pi/3,\pi/2)$ and strictly improves upon it across the entire interval $\alpha \in (0,\pi/3)$; in particular, the symmetric \textsc{Midpoint} rule ($\lambda = 1/2$) lowers the competitive ratio at $\alpha = \pi/4$ from $5/4$ to $\sqrt{3/2} \approx 1.2247$.

In Section~6, we establish a universal lower bound on the competitive ratio of any online algorithm across the entire domain $\alpha \in (0,\pi/2)$ via Ces\`aro averaging on alternating geometric request sequences, characterized by a single quartic root $s^*(\alpha)$. We also show by convexity that randomization against an oblivious adversary cannot improve upon the optimal deterministic competitive ratio, bracketing the optimal competitive ratio at $\alpha = \pi/4$ within $\bigl[\frac{1+\sqrt{2}}{2}, \sqrt{\frac{3}{2}}\bigr] \approx [1.2071, 1.2247]$. Section~7 concludes with open problems.
\section{Baseline Online Algorithms}\label{sec:basic}

In this section we introduce three natural memoryless online strategies and establish their exact competitive ratios.

\begin{description}
\item[\SU] In this algorithm, the drone always stays on the $y$-axis (with $x = 0$) no matter which side of the origin the next request targets. At the end of step $i+1$, the coverage interval for the algorithm is $I(0, y_{i+1}) = [-y_{i+1} \tan \alpha, \; y_{i+1} \tan \alpha]$. In coverage space, it always stays on the diagonal $u=v$ with $u_i = v_i = y_i \tan \alpha = m_i$.

\item[\ConeApex] The algorithm always follows a trajectory that goes from one cone apex to the next on consecutive requests: in step $i+1$, it traverses a flank of the cone $K_{\alpha}(a_i, b_i)$ to reach the apex $T_{i+1}$ of the next cone. In coverage space, the trajectory is piece-wise, axis-parallel (see, e.g., Figure \ref{fig:cone_apex_trajectory}). Hence, along each segment of the path in coverage space, either $u$ increases and $v$ stays unchanged or $u$ stays unchanged and $v$ increases.

\item[\Greedy] During step $i+1$, the algorithm moves the drone greedily from its current feasible position $P_{i}$ to the point that is closest in Euclidean distance to the feasibility cone $K_\alpha(a_{i+1}, b_{i+1})$. It follows that when $\alpha \le \pi/4$, \Greedy will behave exactly like \ConeApex since the apex of the cone $K_\alpha(a_{i+1},b_{i+1})$ is closest to $P_{i}$, and by induction, \Greedy will move from apex to apex just like \ConeApex. When $\alpha > \pi/4$, the scanning angle $2 \alpha$ becomes obtuse, and \Greedy will move to the point on the cone $K_\alpha(a_{i+1}, b_{i+1})$ along a trajectory that is perpendicular to the nearest flank of that cone (in a manner similar to what \OPT does when $\alpha > \pi/4$ and $r < - \cos(2 \alpha))$.
\end{description}

Figure \ref{fig:all_baselines_trajectory} shows how each of these strategies move in position space and in coverage space for the input sequence $X=(0, 1, -2, 5)$ using a drone with half-scanning angle $\alpha=\pi/3$. We overlay \OPT's trajectory in the same figure. Observe that all the three baseline algorithms are instances of a \emph{fixed-angle steering rule}: for a fixed $\beta \in [0, \alpha]$, the drone stays put on covered requests, steers at angle $-\beta$ on uncovered left requests, and steers at angle $+\beta$ on uncovered right requests. In particular, \SU steers with angle $\beta = 0$ and \ConeApex steers at angle $\alpha$ if moving right or angle $-\alpha$ if moving left. \Greedy's steering angle depends on $\alpha$: if $\alpha \le \pi/4$, it steers exactly like \ConeApex, but when $\alpha > \pi/4$, its steering angle is $\pm (\pi/2 - \alpha)$ depending on whether it is moving right or left. This motivates the results in the next subsection, results that are applicable to the entire family of online algorithms, $\mathcal{A}_\beta$, which is  parametrized by a fixed steering angle $\beta \in [0, \alpha]$.

\begin{figure}[htbp]
\centering
\begin{tikzpicture}[
    scale=1.0,
    >=Stealth,
    point/.style={circle, fill=black, inner sep=1.2pt},
    greedypt/.style={circle, fill=blue!80!black, inner sep=1.5pt},
    supt/.style={rectangle, fill=teal!80!black, inner sep=1.6pt},
    capt/.style={diamond, fill=orange!85!black, inner sep=1.4pt},
    optpt/.style={circle, fill=red!80!black, inner sep=1.5pt},
    optline/.style={thick, dashed, red!80!black},
    greedyline/.style={thick, blue!80!black, ->},
    suline/.style={thick, teal!80!black, ->},
    caline/.style={thick, orange!85!black, ->},
    coneline/.style={gray!45, dashed, thin}
]

  % =========================================================
  % Shared Compact Legend
  % =========================================================
  \begin{scope}[shift={(4.5, 6.5)}, font=\footnotesize]
    \draw[suline] (-4.4, 0) -- (-3.8, 0) node[midway, supt] {};
    \node[right] at (-3.8, 0) {\textsc{Straight-Up}};
    \draw[caline] (-1.6, 0) -- (-1.0, 0) node[midway, capt] {};
    \node[right] at (-1.0, 0) {\textsc{Cone-Apex}};
    \draw[greedyline] (1.1, 0) -- (1.7, 0) node[midway, greedypt] {};
    \node[right] at (1.7, 0) {\textsc{Greedy}};
    \draw[optline] (3.3, 0) -- (3.9, 0) node[midway, optpt] {};
    \node[right] at (3.9, 0) {\textsc{Opt}};
  \end{scope}

  % =========================================================
  % SUBFIGURE 1: (x,y) Physical Space (\alpha = 60 deg)
  % =========================================================
  \begin{scope}[shift={(0,0)}]

    % Axes (light y-axis so Straight-Up stands out)
    \draw[->, thick] (-2.25, 0) -- (5.25, 0) node[right] {$x$};
    \draw[->, gray!40, thin] (0, 0) -- (0, 3.35) node[above, black] {$y$};

    % Target Requests on x-axis
    \node[point, label=below:{\footnotesize $x_0$}] at (0,0) {};
    \node[point, label=below:{\footnotesize $x_1$}] at (1,0) {};
    \node[point, label=below:{\footnotesize $x_2$}] at (-2,0) {};
    \node[point, label=below:{\footnotesize $x_3$}] at (5,0) {};

    % Feasibility Cone Boundaries (slope = +/- 1/sqrt(3))
    \draw[coneline] (0, 0) -- (3.0, {3.0/sqrt(3)});
    \draw[coneline] (-2,0) -- (2.5, {4.5/sqrt(3)});
    \draw[coneline] (1,0) -- (-2.2, {3.2/sqrt(3)});
    \draw[coneline] (5,0) -- (0.2, {4.8/sqrt(3)});

    % Straight-Up Positions: (0, sqrt(3)/3), (0, 2*sqrt(3)/3), (0, 5*sqrt(3)/3)
    \node[supt] (S1) at (0, {sqrt(3)/3}) {};
    \node[supt] (S2) at (0, {2*sqrt(3)/3}) {};
    \node[supt, label=left:{\footnotesize $S_3$}] (S3) at (0, {5*sqrt(3)/3}) {};

    % Cone-Apex Positions (Cone Apexes): (0.5, sqrt(3)/6), (-0.5, sqrt(3)/2), (1.5, 7*sqrt(3)/6)
    \node[capt, label=right:{\footnotesize $T_1$}] (T1) at (0.5, {sqrt(3)/6}) {};
    \node[capt, label=left:{\footnotesize $T_2$}] (T2) at (-0.5, {sqrt(3)/2}) {};
    \node[capt, label=right:{\footnotesize $T_3$}] (T3) at (1.5, {7*sqrt(3)/6}) {};

    % True Greedy Positions (Orthogonal Projections)
    \node[greedypt] (P0) at (0, 0) {};
    \node[greedypt, label={[yshift=2pt]above:{\footnotesize $P_1$}}] (P1) at (0.25, {sqrt(3)/4}) {};
    \node[greedypt, label={[yshift=2pt]left:{\footnotesize $P_2$}}] (P2) at (-0.125, {5*sqrt(3)/8}) {};
    \node[greedypt, label=above:{\footnotesize $P_3$}] (P3) at (0.6875, {23*sqrt(3)/16}) {};

    % OPT End Point (Off-Apex Projection onto right flank x + y*sqrt(3) = 5)
    \node[optpt, label={[yshift=0pt]above:{\footnotesize $P_{\text{OPT}}$}}] (POPT) at (1.25, {5*sqrt(3)/4}) {};

    % Trajectories
    \draw[optline] (P0) -- (POPT);
    \draw[suline] (P0) -- (S1);
    \draw[suline] (S1) -- (S2);
    \draw[suline] (S2) -- (S3);
    \draw[caline] (P0) -- (T1);
    \draw[caline] (T1) -- (T2);
    \draw[caline] (T2) -- (T3);
    \draw[greedyline] (P0) -- (P1);
    \draw[greedyline] (P1) -- (P2);
    \draw[greedyline] (P2) -- (P3);

    \node[font=\footnotesize] at (1.5, -0.75) {(a) Physical Space $(x,y)$};
  \end{scope}

  % =========================================================
  % SUBFIGURE 2: (u,v) Coverage Space
  % =========================================================
  \begin{scope}[shift={(6.8,0)}]

    % Axes (u extended to 5.6 to accommodate Straight-Up at (5,5))
    \draw[->, thick] (0, 0) -- (5.6, 0) node[right] {$u$};
    \draw[->, thick] (0, 0) -- (0, 5.65) node[above] {$v$};

    % Target Feasibility Bounds (u >= 2, v >= 5)
    \draw[gray!40, dotted, thick] (2,0) -- (2,5.4) node[above, black, font=\tiny] {$u=2$};
    \draw[gray!40, dotted, thick] (0,5) -- (5.3,5) node[right, black, font=\tiny] {$v=5$};
    \draw[gray!40, dotted, thick] (0,1) -- (5.3,1) node[right, black, font=\tiny] {$v=1$};

    % Origin
    \node[point, label=below left:{\footnotesize $(0,0)$}] (U0) at (0,0) {};

    % Straight-Up Points in (u,v) Space
    \node[supt, label=below right:{\footnotesize $S_1'$}] (SU1) at (1,1) {};
    \node[supt, label=above left:{\footnotesize $S_2'$}] (SU2) at (2,2) {};
    \node[supt, label=above right:{\footnotesize $S_3'$}] (SU3) at (5,5) {};

    % Cone-Apex Points in (u,v) Space
    \node[capt, label=left:{\footnotesize $T_1'$}] (CA1) at (0,1) {};
    \node[capt, label=below right:{\footnotesize $T_2'$}] (CA2) at (2,1) {};
    \node[capt, label=left:{\footnotesize $T_3'$}] (CA3) at (2,5) {};

    % Greedy Points in (u,v) Space
    \node[greedypt, label=above:{\footnotesize $P_1'$}] (U1) at (0.5,1) {};
    \node[greedypt, label=right:{\footnotesize $P_2'$}] (U2) at (2,1.75) {};
    \node[greedypt, label=right:{\footnotesize $P_3'$}] (U3) at (3.625,5) {};

    % OPT End Point in (u,v) Space
    \node[optpt, label={[xshift=-2pt]above right:{\footnotesize $P_{\text{OPT}}'$}}] (UOPT) at (2.5,5) {};

    % Trajectories
    \draw[optline] (U0) -- (UOPT);
    \draw[suline] (U0) -- (SU1);
    \draw[suline] (SU1) -- (SU2);
    \draw[suline] (SU2) -- (SU3);
    \draw[caline] (U0) -- (CA1);
    \draw[caline] (CA1) -- (CA2);
    \draw[caline] (CA2) -- (CA3);
    \draw[greedyline] (U0) -- (U1);
    \draw[greedyline] (U1) -- (U2);
    \draw[greedyline] (U2) -- (U3);

    \node[font=\footnotesize] at (2.8, -0.75) {(b) Coverage Space $(u,v)$};
  \end{scope}

\end{tikzpicture}
\caption{Trajectories of \textsc{Straight-Up}, \textsc{Cone-Apex}, \textsc{Greedy}, and \textsc{Opt} on sequence $(0, 1, -2, 5)$ for $\alpha = \pi/3$.}
\label{fig:all_baselines_trajectory}
\end{figure}

\subsection{Over-Coverage Analysis For $\mathcal{A}_\beta$} \label{subsec:overcoverage-lemma}

\begin{lemma}\label{lem:delta_step}
Let $\mathcal{A}_\beta$ be the online algorithm with fixed steering angle $\beta \in [0, \alpha]$ currently at physical position $(x,y)$. Let $(u,v)$ be its coordinates in coverage space. Then, with
\begin{align}
\sigma &= \frac{\tan \alpha - \tan \beta}{\tan \alpha + \tan \beta}  \label{eq:sigma}
\end{align}
\begin{enumerate}[(a)]
\item if the next request is an active left move that requires expanding the left endpoint of coverage to $-(u + \Delta u)$, then $\mathcal{A}_\beta$ adds \textit{passive right over-coverage}
\[ \Delta v = \sigma \Delta u \]
and arrives at point $(u + \Delta u, \; v + \Delta v)$ in coverage space after traveling a distance of $N_{\alpha} (1, \sigma) \Delta u$.
\item if the next request is an active right move that requires expanding the right endpoint of coverage to $v + \Delta v$, then $\mathcal{A}_\beta$ adds \textit{passive left over-coverage}
\[ \Delta u = \sigma \Delta v \]
and arrives at point $(u + \Delta u, \; v + \Delta v)$ in coverage space after traveling a distance of $N_{\alpha} (1, \sigma) \Delta v$.
\end{enumerate}
\end{lemma}

\begin{proof}
The incremental distance travelled by the algorithm is
\[ \norm{(\Delta x, \; \Delta y)} = N_{\alpha}(\Delta u, \; \Delta v) \]
with $\Delta y > 0$ since the drone must increase altitude to cover a wider interval.
By linearity of~\eqref{eq:uv_xy}, its new position $(u + \Delta u, \; v + \Delta v)$ in coverage space satisfies
\begin{align}
(\Delta u, \; \Delta v) &= (\Delta y \: \tan \alpha - \Delta x, \; \Delta y \: \tan \alpha + \Delta x)  \label{eq:delta_uv_xy}
\end{align}
Define
\[ \sigma = \frac{\tan \alpha - \tan \beta}{\tan \alpha + \tan \beta} \]
We know that $\mathcal{A}_\beta$ steers using the fixed angle $-\beta$ for active left moves and $\beta$ for active right moves.
\begin{enumerate}[(a)]
\item If the next request is an active left move, then $\Delta x = \tan (-\beta) \cdot \Delta y = - \tan \beta \cdot \Delta y$. Substituting this in  ~\eqref{eq:delta_uv_xy}:
\begin{align*}
(\Delta u, \Delta v) &= (\Delta y \: \tan \alpha + \Delta y \: \tan \beta,
\; \Delta y \: \tan \alpha - \Delta y \: \tan \beta ) \\
&= \Delta y \: (\tan \alpha + \tan \beta, \; \tan \alpha - \tan \beta)
\end{align*}
and using the definition of $\sigma$, we get $\Delta v = \sigma \Delta u$.
\item This is proved by symmetry from part (a).
\end{enumerate}
The Euclidean distance traveled is $\norm{(\Delta x,\; \Delta y)} = N_{\alpha}(\Delta u, \; \Delta v)$. The latter value is $N_{\alpha}(1, \sigma) \Delta u $ if the move is an active left move and is $N_{\alpha}(\sigma, 1) \Delta v$ if the move is an active right move. But $N_{\alpha}(1, \sigma) = N_{\alpha}(\sigma, 1)$, which concludes the proof.
\end{proof}

By Lemma \ref{lem:delta_step}, every unit of active expansion on one side generates $\sigma$ units of passive over-coverage on the
opposite side at Euclidean cost $N_{\alpha}(1, \sigma)$, where
\begin{align}
\sigma &= \frac{\tan \alpha - \tan \beta}{\tan \alpha + \tan \beta} \in [0,1] \text{ and } N_{\alpha}(1, \sigma) &= \frac{\sqrt{1 + \sigma^2 + 2c \sigma}}{2 \sin \alpha} \label{eq:sigma-of-beta}
\end{align}
As an illustrative example, observe the trajectory for \Greedy in Figure \ref{fig:all_baselines_trajectory} with $\alpha = \pi/3$. We note that \Greedy goes perpendicularly between coverage cones in position space using a fixed steering angle of $\pi/6$ and hence with over-coverage rate $\sigma = 1/2$. In coverage space, it starts at the origin $(0,0)$ and moves successively to
$P_1'=(0.5,1)$, $P_2'=(2,1.75)$ and $P_3'=(3.625,5)$ in response to inputs $x_1=1,~x_2=-2$ and $x_3=5$. The first move is an active right move with $\Delta v=1$ and by virtue of passive over-coverage on the other side, \Greedy ends up with $\Delta u = \sigma \Delta v = 0.5$ and hence at point $P_1'$. The second move requires active left movement from $u=0.5$ to $u=2$, and the passive over-coverage on the right is $0.5 (2-0.5) = 0.75$ thus taking \Greedy to $P_1^{'} + (1.5, 0.75) = (2, 1.75) = P_2^{'}$ in coverage space. In summary, the active expansions are $\Delta v = 1$, then $\Delta u = 1.5$ and finally $\Delta v = 3.25$, and \Greedy's movement cost on this sequence is $(1 + 1.5 + 3.25)N_{\pi/3}(1, 1/2)$.

The next lemma converts this single-step over-coverage accounting into a tight worst-case bound over arbitrary request
sequences.
\begin{lemma}[Fixed-Angle Over-Coverage Bound]\label{lem:fixed-angle-cost}
Let $\mathcal{A}_\beta$ be the online algorithm with fixed steering angle $\beta \in [0, \alpha]$ and passive over-coverage ratio $\sigma \in [0, 1]$ given by~\eqref{eq:sigma}. For every finite request sequence $X = (x_0 = 0, x_1, \ldots, x_k)$ with final extreme bounds $a_k, b_k$, dominant bound $m = \max(a_k, b_k) > 0$, non-dominant bound $n = \min(a_k, b_k) \ge 0$, and asymmetry ratio $r = n/m \in [0, 1]$,
\begin{equation}\label{eq:fixed-angle-ub}
\mathcal{A}_\beta(X) \le m\big(1 + (1 - \sigma)r\big)\, N_\alpha(1, \sigma).
\end{equation}
Moreover, if $a_k = m$ and $b_k = n$, equality holds in~\eqref{eq:fixed-angle-ub} on the two-step sequence $X = (0, n, -m)$ (and symmetrically on $X = (0, -n, m)$ when $b_k = m$ and $a_k = n$).
\end{lemma}

\begin{proof}
Assume without loss of generality that $a_k = m$ and $b_k = n$, and let $(u, v)$ be the final coverage position of $\mathcal{A}_\beta$. Let $u_h, v_h \ge 0$ be the cumulative active expansions across all left and right moves, respectively. Since active left moves add $\sigma u_h$ to $v$ and active right moves add $\sigma v_h$ to $u$, we have $u = u_h + \sigma v_h$ and $v = v_h + \sigma u_h$. Summing gives $u + v = (1 + \sigma)(u_h + v_h)$, so
\begin{equation}\label{eq:path-uv-sum}
\mathcal{A}_\beta(X) = (u_h + v_h)\, N_\alpha(1, \sigma) = \frac{u + v}{1 + \sigma}\, N_\alpha(1, \sigma).
\end{equation}
If the last positive move serves the dominant left side, then $u = m$ and $v_h \le n$; substituting $u_h = m - \sigma v_h$ and using $1 - \sigma^2 \ge 0$ gives
\[
v = v_h + \sigma(m - \sigma v_h) = \sigma m + (1 - \sigma^2)v_h \le \sigma m + (1 - \sigma^2)n,
\]
whence $u + v \le (1 + \sigma)m + (1 - \sigma^2)n = (1 + \sigma)\big(m + (1 - \sigma)n\big)$. If instead the last positive move serves the non-dominant right side, symmetry yields $u + v \le (1 + \sigma)\big(n + (1 - \sigma)m\big)$, which is smaller by $\sigma(m - n) \ge 0$. Substituting $u + v \le (1 + \sigma)\big(m + (1 - \sigma)n\big)$ and $n = rm$ into~\eqref{eq:path-uv-sum} gives~\eqref{eq:fixed-angle-ub}.

Finally, on $X = (0, n, -m)$, the request $x_1 = n$ moves $\mathcal{A}_\beta$ to $(\sigma n, n)$ at cost $n\,N_\alpha(1, \sigma)$, and since $\sigma n \le n \le m$, the request $x_2 = -m$ requires active expansion $\Delta u = m - \sigma n$ at cost $(m - \sigma n)N_\alpha(1, \sigma)$, totaling $\big(n + m - \sigma n\big)N_\alpha(1, \sigma) = m\big(1 + (1 - \sigma)r\big)N_\alpha(1, \sigma)$.
\end{proof}

\subsection{Competitive Ratios of the Baseline Algorithms}\label{subsec:baseline-ratios}

Throughout this subsection, let $X = (x_0 = 0, x_1, \ldots, x_k)$ denote an arbitrary finite request sequence with dominant final bound $m = \max(a_k, b_k) > 0$, non-dominant bound $n = \min(a_k, b_k) \ge 0$, asymmetry ratio $r = n/m \in [0, 1]$, and $c = \cos(2\alpha)$. (If $m = 0$, every algorithm remains at the origin at zero cost.)

\begin{theorem}[Competitive Ratio of \textsc{Straight-Up}]\label{thm:straight-up}
For every scanning half-angle $\alpha \in (0, \pi/2)$, the competitive ratio of \textsc{Straight-Up} is
\[
\rho(\textsc{Straight-Up}; \alpha) =
\begin{cases}
2\cos\alpha & \text{if } \alpha \in (0, \pi/4], \\[4pt]
\csc\alpha  & \text{if } \alpha \in (\pi/4, \pi/2).
\end{cases}
\]
\end{theorem}

\begin{proof}
\textsc{Straight-Up} uses fixed steering angle $\beta = 0$, which corresponds to $\sigma = 1$ and unit cost $N_\alpha(1, 1) = \cot\alpha$. Lemma~\ref{lem:fixed-angle-cost} immediately gives
\begin{equation}\label{eq:su-ub}
\textsc{Straight-Up}(X) \le m\cot\alpha.
\end{equation}
\begin{itemize}
    \item \emph{Narrow angles ($\alpha \in (0, \pi/4]$, so $c \ge 0$).} By Theorem~\ref{thm:opt}, $\mathrm{OPT}(X) = \frac{m}{2\sin\alpha}\sqrt{1 + r^2 + 2cr}$ for all $r \in [0, 1]$. Dividing~\eqref{eq:su-ub} by $\mathrm{OPT}(X)$ gives
    \[
    \frac{\textsc{Straight-Up}(X)}{\mathrm{OPT}(X)} \le \frac{2\cos\alpha}{\sqrt{1 + r^2 + 2cr}} \le 2\cos\alpha,
    \]
    since $1 + r^2 + 2cr \ge 1$ for $c \ge 0$ and $r \ge 0$.
    \item \emph{Wide angles ($\alpha \in (\pi/4, \pi/2)$, so $c < 0$).} By Theorem~\ref{thm:opt}, if $r \in [0, -c)$ then $\mathrm{OPT}(X) = m\cos\alpha$, giving $\frac{\textsc{Straight-Up}(X)}{\mathrm{OPT}(X)} \le \frac{m\cot\alpha}{m\cos\alpha} = \csc\alpha$. If $r \in [-c, 1]$, then $\mathrm{OPT}(X) = \frac{m}{2\sin\alpha}\sqrt{1 + r^2 + 2cr}$; since $\frac{d}{dr}(1 + r^2 + 2cr) = 2(r + c) \ge 0$ on $[-c, 1]$ and $\sqrt{1 - c^2} = \sin(2\alpha)$, the ratio is non-increasing on $[-c, 1]$ and bounded by its value at $r = -c$:
    \[
    \frac{\textsc{Straight-Up}(X)}{\mathrm{OPT}(X)} \le \frac{2\cos\alpha}{\sqrt{1 - c^2}} = \frac{2\cos\alpha}{\sin(2\alpha)} = \csc\alpha.
    \]
\end{itemize}
For tightness, consider the single-request sequence $X = (0, -1)$ (so $m = 1$, $n = 0$, and $r^* = 0$). Lemma~\ref{lem:fixed-angle-cost} gives $\textsc{Straight-Up}(X) = \cot\alpha$, while Theorem~\ref{thm:opt} gives $\mathrm{OPT}(X) = \frac{1}{2\sin\alpha}$ when $\alpha \le \pi/4$ and $\mathrm{OPT}(X) = \cos\alpha$ when $\alpha > \pi/4$, matching the upper bound in both regimes.
\end{proof}

\begin{theorem}[Competitive Ratio of \textsc{Cone-Apex}]\label{thm:cone-apex}
For every scanning half-angle $\alpha \in (0, \pi/2)$, the competitive ratio of \textsc{Cone-Apex} is
\[
\rho(\textsc{Cone-Apex}; \alpha) = \sec\alpha.
\]
\end{theorem}

\begin{proof}
\textsc{Cone-Apex} uses fixed steering angle $\beta = \alpha$, which corresponds to $\sigma = 0$ and unit cost $N_\alpha(1, 0) = \frac{1}{2\sin\alpha}$. Lemma~\ref{lem:fixed-angle-cost} yields
\begin{equation}\label{eq:ca-ub}
\textsc{Cone-Apex}(X) \le \frac{m(1 + r)}{2\sin\alpha}.
\end{equation}
Whenever $\mathrm{OPT}(X)$ lies at the cone apex—namely for all $r \in [0, 1]$ when $\alpha \le \pi/4$, and for $r \in [-c, 1]$ when $\alpha > \pi/4$—Theorem~\ref{thm:opt} and~\eqref{eq:ca-ub} give
\[
\frac{\textsc{Cone-Apex}(X)}{\mathrm{OPT}(X)} \le g(r) := \frac{1 + r}{\sqrt{1 + r^2 + 2cr}}.
\]
Observe that $1 + r^2 + 2cr = {\cos}^2 \alpha (1 + r)^2 + {\sin}^2 \alpha (1 - r)^2 \geq {\cos}^2 \alpha (1 + r)^2$. Hence, $g(r) \leq \sec \alpha$.
Differentiating yields $g'(r) = \frac{(1 - c)(1 - r)}{(1 + r^2 + 2cr)^{3/2}} \ge 0$ for all $r \in [0, 1]$ (as $c = \cos(2\alpha) < 1$), so $g(r)$ is maximized at $r = 1$:
\[
g(r) \le g(1) = \frac{2}{\sqrt{2 + 2c}} = \frac{2}{2\cos\alpha} = \sec\alpha.
\]
If $\alpha > \pi/4$ and $r \in [0, -c)$, Theorem~\ref{thm:opt} gives $\mathrm{OPT}(X) = m\cos\alpha$, so
\[
\frac{\textsc{Cone-Apex}(X)}{\mathrm{OPT}(X)} \le \frac{1 + r}{2\sin\alpha\cos\alpha} = \frac{1 + r}{\sin(2\alpha)} < \frac{1 - c}{\sin(2\alpha)} = \frac{2\sin^2\alpha}{2\sin\alpha\cos\alpha} = \tan\alpha < \sec\alpha.
\]
For tightness, on the symmetric two-request sequence $X = (0, 1, -1)$ ($m = n = 1$, $r^* = 1$), Lemma~\ref{lem:fixed-angle-cost} gives $\textsc{Cone-Apex}(X) = \frac{1}{\sin\alpha}$ while Theorem~\ref{thm:opt} gives $\mathrm{OPT}(X) = \cot\alpha$, achieving ratio $\sec\alpha$ for every $\alpha \in (0, \pi/2)$.
\end{proof}

\begin{theorem}[Competitive Ratio of \textsc{Greedy}]\label{thm:greedy}\label{sec:greedy}
For every scanning half-angle $\alpha \in (0, \pi/2)$, the competitive ratio of \textsc{Greedy} is
\[
\rho(\textsc{Greedy}; \alpha) =
\begin{cases}
\sec\alpha & \text{if } \alpha \in (0, \pi/4], \\[4pt]
\sqrt{2 - \cos^2(2\alpha)} & \text{if } \alpha \in (\pi/4, \pi/2).
\end{cases}
\]
\end{theorem}

\begin{proof}
For $\alpha \in (0, \pi/4]$, \textsc{Greedy} coincides identically with \textsc{Cone-Apex}, so $\rho(\textsc{Greedy}; \alpha) = \sec\alpha$ by Theorem~\ref{thm:cone-apex}.

Assume now that $\alpha \in (\pi/4, \pi/2)$, so $c = \cos(2\alpha) \in (-1, 0)$. As noted earlier, \textsc{Greedy} steers at fixed angle $\beta = \frac{\pi}{2} - \alpha$, giving $\sigma = -c$ and unit cost $N_\alpha(1, -c) = \cos\alpha$. Lemma~\ref{lem:fixed-angle-cost} therefore yields
\begin{equation}\label{eq:greedy-ub}
\textsc{Greedy}(X) \le m\big(1 + (1 + c)r\big)\cos\alpha.
\end{equation}
We compare~\eqref{eq:greedy-ub} with Theorem~\ref{thm:opt} across both regimes of $r$:
\begin{itemize}
    \item \emph{Off-apex regime ($0 \le r < -c$).} Here $\mathrm{OPT}(X) = m\cos\alpha$, and since $1 + c = 2\cos^2\alpha > 0$,
    \[
    \frac{\textsc{Greedy}(X)}{\mathrm{OPT}(X)} \le 1 + (1 + c)r < 1 - c - c^2.
    \]
    \item \emph{Apex regime ($-c \le r \le 1$).} Here $\mathrm{OPT}(X) = \frac{m}{2\sin\alpha}\sqrt{1 + r^2 + 2cr}$, so~\eqref{eq:greedy-ub} gives
    \[
    \frac{\textsc{Greedy}(X)}{\mathrm{OPT}(X)} \le f(r) := \sin(2\alpha)\,\frac{1 + (1 + c)r}{\sqrt{1 + r^2 + 2cr}}.
    \]
    Differentiating with respect to $r$ yields
    \[
    f'(r) = \sin(2\alpha)\,\frac{1 - r(1 - c - c^2)}{(1 + r^2 + 2cr)^{3/2}}.
    \]
    For $c \in (-1, 0)$, we have $1 - c - c^2 = \frac{5}{4} - \big(c + \frac{1}{2}\big)^2 > 1$, so $f(r)$ attains its unique maximum on $\mathbb{R}$ at $r^* = \frac{1}{1 - c - c^2}$. Because
    \[
    r^* - (-c) = \frac{(1 - c)(1 + c)^2}{1 - c - c^2} > 0
    \qquad\text{and}\qquad
    1 - r^* = \frac{-c(1 + c)}{1 - c - c^2} > 0,
    \]
    we have $r^* \in (-c, 1)$. Evaluating $f(r^*)$ using $\sin(2\alpha) = \sqrt{1 - c^2}$ gives
    \[
    f(r^*) = \sqrt{1 - c^2}\,\frac{\frac{2 - c^2}{1 - c - c^2}}{\frac{\sqrt{(1 - c^2)(2 - c^2)}}{1 - c - c^2}} = \sqrt{2 - c^2} = \sqrt{2 - \cos^2(2\alpha)}.
    \]
\end{itemize}
Since $f(-c) = 1 - c - c^2 < f(r^*)$, the upper bound $\sqrt{2 - \cos^2(2\alpha)}$ holds for all $r \in [0, 1]$. Finally, on the two-request sequence $X = (0, r^*, -1)$ ($m = 1$, $n = r^* \in (-c, 1)$), Lemma~\ref{lem:fixed-angle-cost} gives $\textsc{Greedy}(X) = \big(1 + (1 + c)r^*\big)\cos\alpha$ while Theorem~\ref{thm:opt} gives $\mathrm{OPT}(X) = \frac{1}{2\sin\alpha}\sqrt{1 + (r^*)^2 + 2cr^*}$, attaining $\sqrt{2 - \cos^2(2\alpha)}$ with equality.
\end{proof}
\section{The $\beta$-\textsc{Hedge} Algorithm}\label{sec:beta-hedge}

As shown in Section~\ref{sec:basic}, \textsc{Straight-Up} performs best on symmetric sequences ($r \approx 1$), whereas \textsc{Cone-Apex} excels on highly asymmetric, one-sided sequences ($r \approx 0$). Moreover, \textsc{Straight-Up} ($\beta = 0$), \textsc{Cone-Apex} ($\beta = \alpha$), and wide-angle \textsc{Greedy} ($\beta = \frac{\pi}{2} - \alpha$) are all isolated instances of the fixed-angle steering rule $\mathcal{A}_\beta$ from Section~\ref{subsec:overcoverage-lemma}: whenever a request falls outside the current coverage interval $[-u, v]$, steer in a straight line at angle $-\beta$ on uncovered left requests and $+\beta$ on uncovered right requests until the request is covered.

Allowing $\beta \in [0, \alpha]$ to vary continuously defines the \emph{fixed-angle $\beta$-\textsc{Hedge}} algorithm $\mathcal{A}_\beta$, which interpolates directly between \textsc{Straight-Up} ($\beta = 0$) and \textsc{Cone-Apex} ($\beta = \alpha$) across all scanning half-angles $\alpha \in (0, \pi/2)$. (For $\alpha \le \pi/4$, the endpoint $\beta = \alpha$ also coincides with \textsc{Greedy}; for $\alpha > \pi/4$, wide-angle \textsc{Greedy} corresponds to the interior angle $\beta = \frac{\pi}{2} - \alpha \in (0, \alpha)$.)

Our analysis proceeds in two stages. First, in Section~\ref{subsec:fixed-beta}, we combine Lemma~\ref{lem:fixed-angle-cost} with Theorem~\ref{thm:opt} to determine the exact competitive ratio of $\beta$-\textsc{Hedge} for any fixed $\beta \in [0, \alpha]$. Second, in Section~\ref{subsec:opt-beta}, we optimize over $\beta \in [0, \alpha]$ to obtain, for every $\alpha \in (0, \pi/2)$, the optimal steering angle $\beta^*(\alpha)$.

\subsection{Competitive Ratio for a Fixed Hedging Angle $\beta$}\label{subsec:fixed-beta}

For a fixed steering angle $\beta \in [0, \alpha]$, recall from~\eqref{eq:sigma-of-beta} that the passive over-coverage ratio is
\begin{equation}\label{eq:sigma-beta-hedge}
\sigma = \frac{\tan\alpha - \tan\beta}{\tan\alpha + \tan\beta} = \frac{\sin(\alpha - \beta)}{\sin(\alpha + \beta)} \in [0, 1],
\end{equation}
and define the complementary hedging weight
\begin{equation}\label{eq:lambda-beta-hedge}
\lambda = 1 - \sigma = \frac{2\tan\beta}{\tan\alpha + \tan\beta} = \frac{2\sin\beta\cos\alpha}{\sin(\alpha + \beta)} \in [0, 1].
\end{equation}
Thus $\lambda = 0$ ($\sigma = 1$) corresponds to \textsc{Straight-Up} ($\beta = 0$), while $\lambda = 1$ ($\sigma = 0$) corresponds to \textsc{Cone-Apex} ($\beta = \alpha$). Using $c = \cos(2\alpha) = 1 - 2\sin^2\alpha$ together with~\eqref{eq:sigma-beta-hedge} and~\eqref{eq:lambda-beta-hedge}, the unit expansion cost $N_\alpha(1, \sigma)$ simplifies for all $\beta \in [0, \alpha]$ to
\begin{align}
N_\alpha(1, \sigma) &= \frac{\sqrt{1 + \sigma^2 + 2c\sigma}}{2\sin\alpha}
= \frac{\sqrt{(1 - \sigma)^2 + 2(1 + c)\sigma}}{2\sin\alpha} \nonumber \\
&= \frac{\sqrt{\lambda^2 + 4\cos^2\alpha\,(1 - \lambda)}}{2\sin\alpha}
= \frac{\cos\alpha}{\sin(\alpha + \beta)},\label{eq:unit-cost-beta}
\end{align}
where the last step uses $\lambda^2 + 4\cos^2\alpha\,(1 - \lambda) = \frac{4\cos^2\alpha\big(\sin^2\beta + \sin(\alpha-\beta)\sin(\alpha+\beta)\big)}{\sin^2(\alpha+\beta)} = \frac{4\sin^2\alpha\cos^2\alpha}{\sin^2(\alpha+\beta)}$.

\begin{theorem}\label{thm:fixed-beta}
For any scanning half-angle $\alpha \in (0, \pi/2)$ and any fixed steering angle $\beta \in [0, \alpha]$, let $c = \cos(2\alpha)$ and $\lambda = \frac{2\tan\beta}{\tan\alpha + \tan\beta} \in [0, 1]$. Define
\[
r^* = \max\!\left(0,\, \frac{\lambda - c}{1 - c\lambda}\right) \in [0, 1].
\]
Then the competitive ratio of fixed-angle $\beta$-\textsc{Hedge} is tightly given by
\begin{equation}\label{eq:cr-fixed-beta}
\rho(\beta\text{-}\textsc{Hedge}; \alpha, \beta) = \frac{\sin(2\alpha)}{\sin(\alpha + \beta)} \cdot \frac{1 + \lambda r^*}{\sqrt{1 + (r^*)^2 + 2c r^*}}.
\end{equation}
In particular, whenever $\lambda \ge \max(0, c)$, this simplifies to
\begin{equation}\label{eq:cr-fixed-beta-interior}
\rho(\beta\text{-}\textsc{Hedge}; \alpha, \beta) = \frac{\sqrt{1 - 2c\lambda + \lambda^2}}{\sin(\alpha + \beta)}.
\end{equation}
\end{theorem}

\begin{proof}
Let $X = (x_0 = 0, x_1, \dots, x_k)$ be any finite request sequence with dominant bound $m = \max(a_k, b_k) > 0$, non-dominant bound $n = \min(a_k, b_k) \ge 0$, and asymmetry ratio $r = n/m \in [0, 1]$. Substituting $1 - \sigma = \lambda$ and~\eqref{eq:unit-cost-beta} into Lemma~\ref{lem:fixed-angle-cost} immediately yields
\begin{equation}\label{eq:alg-ub-beta}
\beta\text{-}\textsc{Hedge}(X) \le m(1 + \lambda r)\frac{\cos\alpha}{\sin(\alpha + \beta)}.
\end{equation}

We now bound $\beta\text{-}\textsc{Hedge}(X) / \mathrm{OPT}(X)$ over $r \in [0, 1]$ using Theorem~\ref{thm:opt}:
\begin{itemize}
    \item \emph{Apex regime ($r \in [\max(0, -c), 1]$).} Here $\mathrm{OPT}(X) = \frac{m}{2\sin\alpha}\sqrt{1 + r^2 + 2cr}$, so~\eqref{eq:alg-ub-beta} gives
    \[
    \frac{\beta\text{-}\textsc{Hedge}(X)}{\mathrm{OPT}(X)} \le f_{\alpha,\beta}(r) := \frac{\sin(2\alpha)}{\sin(\alpha + \beta)} \cdot \frac{1 + \lambda r}{\sqrt{1 + r^2 + 2cr}}.
    \]
    Differentiating $g(r) = \frac{1 + \lambda r}{\sqrt{1 + r^2 + 2cr}}$ with respect to $r$ yields
    \[
    g'(r) = \frac{(\lambda - c) - r(1 - c\lambda)}{(1 + r^2 + 2cr)^{3/2}}.
    \]
    Since $|c| < 1$ and $\lambda \in [0, 1]$, we have $1 - c\lambda > 0$, so $g(r)$ increases for $r < \tilde{r}$ and decreases for $r > \tilde{r}$, where $\tilde{r} = \frac{\lambda - c}{1 - c\lambda}$. Moreover, $1 - \tilde{r} = \frac{(1 + c)(1 - \lambda)}{1 - c\lambda} \ge 0$ (as $1 + c = 2\cos^2\alpha > 0$ and $\lambda \le 1$), so $\tilde{r} \le 1$ always holds:
    \begin{itemize}
        \item If $c < 0$ ($\alpha > \pi/4$), then $\tilde{r} - (-c) = \frac{\lambda(1 - c^2)}{1 - c\lambda} \ge 0$, so $r^* = \tilde{r} \in [-c, 1]$ maximizes $f_{\alpha,\beta}(r)$ on $[-c, 1]$.
        \item If $c \ge 0$ ($\alpha \le \pi/4$) and $\lambda \ge c$, then $r^* = \tilde{r} \in [0, 1]$ maximizes $f_{\alpha,\beta}(r)$ on $[0, 1]$.
        \item If $c \ge 0$ and $\lambda < c$, then $\tilde{r} < 0$, so $g'(r) < 0$ on $[0, 1]$ and $f_{\alpha,\beta}(r)$ is maximized on $[0, 1]$ at $r^* = 0$.
    \end{itemize}
    \item \emph{Off-apex regime ($c < 0$ and $r \in [0, -c)$).} Here $\mathrm{OPT}(X) = m\cos\alpha$, so~\eqref{eq:alg-ub-beta} bounds the ratio by $\frac{1 + \lambda r}{\sin(\alpha + \beta)} \le \frac{1 - c\lambda}{\sin(\alpha + \beta)} = f_{\alpha,\beta}(-c) \le f_{\alpha,\beta}(r^*)$, since $\lambda \ge 0$ and $\sqrt{1 - c^2} = \sin(2\alpha)$.
\end{itemize}
Thus $\beta\text{-}\textsc{Hedge}(X) / \mathrm{OPT}(X) \le f_{\alpha,\beta}(r^*)$ for every request sequence $X$.

For tightness, Lemma~\ref{lem:fixed-angle-cost} guarantees that $\beta\text{-}\textsc{Hedge}(X) = (1 + \lambda r^*)\frac{\cos\alpha}{\sin(\alpha + \beta)}$ on the single-request sequence $X = (0, -1)$ when $r^* = 0$, and on the two-request sequence $X = (0, r^*, -1)$ when $r^* > 0$. Since $r^* \ge \max(0, -c)$ in both cases, Theorem~\ref{thm:opt} gives $\mathrm{OPT}(X) = \frac{1}{2\sin\alpha}\sqrt{1 + (r^*)^2 + 2cr^*}$, attaining~\eqref{eq:cr-fixed-beta} with equality.

Finally, whenever $\lambda \ge \max(0, c)$, substituting $r^* = \frac{\lambda - c}{1 - c\lambda}$ gives
\[
1 + \lambda r^* = \frac{1 - 2c\lambda + \lambda^2}{1 - c\lambda}
\qquad\text{and}\qquad
1 + (r^*)^2 + 2c r^* = \frac{(1 - c^2)(1 - 2c\lambda + \lambda^2)}{(1 - c\lambda)^2}.
\]
Plugging these into~\eqref{eq:cr-fixed-beta} and cancelling $\sqrt{1 - c^2} = \sin(2\alpha)$ yields~\eqref{eq:cr-fixed-beta-interior}.
\end{proof}

\subsection{Choosing the Optimal Hedging Angle $\beta^*(\alpha)$}\label{subsec:opt-beta}

We now optimize the competitive ratio~\eqref{eq:cr-fixed-beta} over $\beta \in [0, \alpha]$ for each scanning half-angle $\alpha \in (0, \pi/2)$. Geometrically, adjusting $\beta$ balances unit-step efficiency against opposite-side over-coverage:
\begin{itemize}
    \item Near \textsc{Cone-Apex} ($\beta = \alpha - \varepsilon$), tilting slightly inward raises the unit step cost on the symmetric witness $X = (0, 1, -1)$ at rate $\cot(2\alpha)$ while granting passive over-coverage savings at rate $\frac{1}{2}\csc(2\alpha)$; the savings beat the slope penalty if and only if $\cos(2\alpha) < \frac{1}{2}$, i.e., $\alpha > \pi/6$.
    \item Near \textsc{Straight-Up} ($\beta = \varepsilon$ for $\alpha > \pi/4$), tilting slightly outward lowers the dominant-side unit cost at rate $\cot\alpha$ while forcing an opposite-side response on $X = (0, -c, -1)$ at rate $-2c\cot\alpha$; tilting helps if and only if $1 + 2c > 0$, i.e., $\alpha < \pi/3$.
\end{itemize}
The next theorem shows that these two local thresholds govern the global optimum across $(0, \pi/2)$. We denote by $\beta^*$-\textsc{Hedge} the algorithm that selects the optimal fixed steering angle $\beta^*(\alpha)$ in advance from the known scanning half-angle $\alpha$.

\begin{theorem}\label{thm:opt-beta}
For every scanning half-angle $\alpha \in (0, \pi/2)$, the unique steering angle $\beta^*(\alpha) \in [0, \alpha]$ minimizing $\rho(\beta\text{-}\textsc{Hedge}; \alpha, \beta)$ is
\begin{equation}\label{eq:beta-star}
\beta^*(\alpha) =
\begin{cases}
\alpha & \text{if } \alpha \in (0, \pi/6], \\[4pt]
\arctan\!\big(2\sin(2\alpha)\big) - \alpha & \text{if } \alpha \in (\pi/6, \pi/3), \\[4pt]
0 & \text{if } \alpha \in [\pi/3, \pi/2),
\end{cases}
\end{equation}
yielding the optimal fixed-angle competitive ratio
\begin{equation}\label{eq:cr-beta-star}
\rho(\beta^*\text{-}\textsc{Hedge}; \alpha) =
\begin{cases}
\sec\alpha & \text{if } \alpha \in (0, \pi/6], \\[4pt]
\sin(2\alpha) + \dfrac{1}{4\sin(2\alpha)} & \text{if } \alpha \in (\pi/6, \pi/3), \\[8pt]
\csc\alpha & \text{if } \alpha \in [\pi/3, \pi/2).
\end{cases}
\end{equation}
In particular, $\beta^*$-\textsc{Hedge} coincides with \textsc{Cone-Apex} on $(0, \pi/6]$, coincides with \textsc{Straight-Up} on $[\pi/3, \pi/2)$, and achieves its worst-case peak $\rho(\beta^*\text{-}\textsc{Hedge}; \pi/4) = \frac{5}{4}$ at $\alpha = \pi/4$ (where $\beta^*(\pi/4) = \arctan(2) - \frac{\pi}{4} = \arctan\frac{1}{3}$).
\end{theorem}

\begin{proof}
Fix $\alpha \in (0, \pi/2)$ and set $c = \cos(2\alpha) \in (-1, 1)$ and $d = \sin(2\alpha) \in (0, 1]$. To express the squared interior ratio~\eqref{eq:cr-fixed-beta-interior} as a polynomial, introduce the total steering angle $\theta = \alpha + \beta \in [\alpha, 2\alpha]$ and its cotangent $z = \cot\theta$. Expanding $\sin(\theta - \alpha) = \sin\theta\cos\alpha - \cos\theta\sin\alpha$ in~\eqref{eq:lambda-beta-hedge} expresses the hedging weight $\lambda$ linearly in $z$:
\begin{equation}\label{eq:lambda-of-z}
\lambda = \frac{2\sin(\theta - \alpha)\cos\alpha}{\sin\theta}
= 2\cos^2\alpha - 2\sin\alpha\cos\alpha\cot\theta
= 1 + c - dz.
\end{equation}
Note that $\lambda$ increases strictly from $0$ to $1$ as $\beta$ increases from $0$ to $\alpha$.

Consider first the interior-branch function from~\eqref{eq:cr-fixed-beta-interior},
\[
H_\alpha(\beta) := \frac{\sqrt{1 - 2c\lambda + \lambda^2}}{\sin(\alpha + \beta)}
= \frac{\sqrt{(\lambda - c)^2 + (1 - c^2)}}{\sin\theta},
\]
extended smoothly across $\beta \in [0, \alpha]$. Substituting $\lambda - c = 1 - dz$ from~\eqref{eq:lambda-of-z} together with $1 - c^2 = d^2$ and $\csc^2\theta = 1 + z^2$ gives
\begin{equation}\label{eq:H-alpha-sq}
H_\alpha(\beta)^2 = (1 + z^2)\big((1 - dz)^2 + d^2\big).
\end{equation}
Differentiating~\eqref{eq:H-alpha-sq} with respect to $z$ factors cleanly as
\begin{equation}\label{eq:dH-dz}
\frac{d}{dz}\big(H_\alpha^2\big) = 2(2dz - 1)(dz^2 - z + d).
\end{equation}
We claim that $dz^2 - z + d > 0$ for all $\theta \in [\alpha, 2\alpha]$. Indeed,
\[
dz^2 - z + d = \frac{d(\cos^2\theta + \sin^2\theta) - \sin\theta\cos\theta}{\sin^2\theta}
= \frac{2d - \sin(2\theta)}{2\sin^2\theta}.
\]
If $\theta \ge \pi/2$, then $2\theta \in [\pi, 2\pi)$ so $\sin(2\theta) \le 0 < 2d$. If $\theta < \pi/2$, then $0 \le 2\beta < \pi - 2\alpha$ gives $\sin(2\beta) \le \sin(\pi - 2\alpha) = \sin(2\alpha) = d$; hence $\sin(2\theta) = d\cos(2\beta) + c\sin(2\beta)$ satisfies $\sin(2\theta) \le d(1 + c) < 2d$ when $c \ge 0$, and $\sin(2\theta) \le d\cos(2\beta) \le d < 2d$ when $c < 0$. Thus $dz^2 - z + d > 0$ throughout $[\alpha, 2\alpha]$.

Since $\frac{dz}{d\beta} = -(1 + z^2) < 0$, Equation~\eqref{eq:dH-dz} implies that $\frac{d}{d\beta}\big(H_\alpha(\beta)^2\bigr)$ has the exact sign of $1 - 2d\cot\theta$. Because $\cot\theta$ is strictly decreasing on $(0, \pi)$, $2d\cot\theta$ decreases strictly across $\theta \in [\alpha, 2\alpha]$ between its endpoint values:
\[
2d\cot\alpha = 4\cos^2\alpha
\qquad\text{and}\qquad
2d\cot(2\alpha) = 2\cos(2\alpha).
\]
This partitions $\alpha \in (0, \pi/2)$ into three regimes for the unconstrained minimizer of $H_\alpha$ on $[0, \alpha]$:
\begin{enumerate}
    \item \emph{Narrow angles ($\alpha \in (0, \pi/6]$).} Here $2\cos(2\alpha) \ge 1$, so $2d\cot\theta \ge 1$ on $[\alpha, 2\alpha]$ and $H_\alpha(\beta)$ is strictly decreasing on $[0, \alpha]$, attaining its unique minimum at $\beta = \alpha$.
    \item \emph{Intermediate angles ($\alpha \in (\pi/6, \pi/3)$).} Here $2\cos(2\alpha) < 1 < 4\cos^2\alpha$, so $H_\alpha(\beta)$ decreases strictly up to the unique stationary point $\theta^* \in (\alpha, 2\alpha)$ satisfying $\cot\theta^* = \frac{1}{2d}$ (equivalently, $\tan\theta^* = 2\sin(2\alpha)$, so $\beta^*(\alpha) = \arctan(2\sin(2\alpha)) - \alpha \in (0, \alpha)$) and increases strictly thereafter.
    \item \emph{Wide angles ($\alpha \in [\pi/3, \pi/2)$).} Here $4\cos^2\alpha \le 1$, so $2d\cot\theta \le 1$ on $[\alpha, 2\alpha]$ and $H_\alpha(\beta)$ is strictly increasing on $[0, \alpha]$, attaining its unique minimum at $\beta = 0$.
\end{enumerate}

It remains to verify that minimizing $H_\alpha(\beta)$ also minimizes the true competitive ratio $\rho(\beta\text{-}\textsc{Hedge}; \alpha, \beta)$ from Theorem~\ref{thm:fixed-beta}. By Theorem~\ref{thm:fixed-beta}, $\rho(\beta\text{-}\textsc{Hedge}; \alpha, \beta) = H_\alpha(\beta)$ whenever $\lambda \ge \max(0, c)$, which holds for all $\beta \in [0, \alpha]$ when $\alpha \ge \pi/4$ (since $c \le 0$). For $\alpha \in (0, \pi/4)$, we have $c > 0$, and on the complementary interval where $\lambda < c$, Theorem~\ref{thm:fixed-beta} gives $r^* = 0$ and
\[
\rho(\beta\text{-}\textsc{Hedge}; \alpha, \beta) = \frac{\sin(2\alpha)}{\sin(\alpha + \beta)},
\]
which is strictly decreasing in $\beta$ since $\alpha + \beta < 2\alpha < \pi/2$. For $\alpha \in (0, \pi/6]$, $\rho(\beta\text{-}\textsc{Hedge}; \alpha, \beta)$ therefore decreases strictly across both branches and is uniquely minimized at $\beta^*(\alpha) = \alpha$. For $\alpha \in (\pi/6, \pi/4)$, evaluating~\eqref{eq:lambda-of-z} at $\cot\theta^* = \frac{1}{2d}$ gives $\lambda^* = 1 + c - d\cot\theta^* = c + \frac{1}{2} > c$, so the stationary point $\beta^*(\alpha)$ lies strictly on the interior branch $\lambda > c$ and remains the unique global minimizer.

Finally, substituting $\beta^*(\alpha)$ into~\eqref{eq:cr-fixed-beta-interior} yields the three competitive ratio formulas in~\eqref{eq:cr-beta-star}: at $\beta^* = \alpha$ ($\lambda = 1$), $H_\alpha(\alpha) = \frac{\sqrt{2 - 2c}}{\sin(2\alpha)} = \frac{2\sin\alpha}{2\sin\alpha\cos\alpha} = \sec\alpha$; at $\cot\theta^* = \frac{1}{2d}$, Equation~\eqref{eq:H-alpha-sq} gives
\[
H_\alpha(\beta^*)^2 = \left(1 + \frac{1}{4d^2}\right)\left(\frac{1}{4} + d^2\right) = \left(d + \frac{1}{4d}\right)^2,
\]
so $H_\alpha(\beta^*) = \sin(2\alpha) + \frac{1}{4\sin(2\alpha)}$; and at $\beta^* = 0$ ($\lambda = 0$), $H_\alpha(0) = \frac{1}{\sin\alpha} = \csc\alpha$.
\end{proof}
\section{State-Based Linear Interpolation}\label{sec:state-interpolation}

Fixed-angle $\beta$-\textsc{Hedge} is inherently path-dependent: whenever a request arrives on the non-dominant side, steering at angle $\pm\beta$ pushes the dominant coverage coordinate further outward and retains that excess cushion on all subsequent steps. In this section, we study a memoryless alternative that depends only on the current request extremes $[-a, b]$ and interpolates linearly in physical space between the cone apex
\[
T_\alpha(a, b) = \left(\frac{b - a}{2},\, \frac{a + b}{2}\cot\alpha\right)
\qquad\text{and}\qquad
S_\alpha(a, b) = \big(0,\, \max(a, b)\cot\alpha\big).
\]
Specifically, for any fixed weight $\lambda \in [0, 1]$, define the state-based target position
\begin{equation}\label{eq:p-lambda-def}
P_\lambda(a, b) = \lambda T_\alpha(a, b) + (1 - \lambda) S_\alpha(a, b),
\end{equation}
and let $\mathcal{A}_\lambda$ be the deterministic online algorithm that moves in a straight line to $P_\lambda(a_i, b_i)$ after each request $x_i$ (staying put when $a_i, b_i$ do not change). Because both $T_\alpha(a, b)$ and $S_\alpha(a, b)$ lie in the feasibility cone $K_{\alpha}(a_i, b_i)$, convexity of $K_{\alpha}(a_i, b_i)$ guarantees that $P_\lambda(a_i, b_i) \in K_{\alpha}(a_i, b_i)$ for every $\lambda \in [0, 1]$.

Writing $P_\lambda(a, b)$ in coverage coordinates $(u, v) = (y\tan\alpha - x,\, y\tan\alpha + x)$ for $a = m \ge b = n$ exposes how $\mathcal{A}_\lambda$ differs from $\beta$-\textsc{Hedge}:
\begin{equation}\label{eq:p-lambda-uv}
\big(u(m, n),\, v(m, n)\big) = \big(m,\, (1 - \lambda)m + \lambda n\big) = \big(m,\, \sigma m + (1 - \sigma)n\big), \qquad \sigma = 1 - \lambda.
\end{equation}
When the dominant bound $m$ grows by $\Delta m > 0$ (with $m \ge n$), $\mathcal{A}_\lambda$ shifts by $(\Delta u, \Delta v) = (\Delta m, \sigma\Delta m)$, in the exact same direction and unit cost as $\beta$-\textsc{Hedge} with passive over-coverage ratio $\sigma = 1 - \lambda$. However, when the non-dominant bound $n$ grows by $\Delta n > 0$ (with $n \le m$), $\mathcal{A}_\lambda$ shifts by $(\Delta u, \Delta v) = (0, \lambda\Delta n)$, keeping $u = m$ pinned and shedding excess altitude rather than accumulating dominant-side over-coverage.

\subsection{The \textsc{Midpoint} Algorithm at $\alpha = \pi/4$}\label{subsec:midpoint}

Before analyzing the full family $\mathcal{A}_\lambda$, consider the symmetric special case $\lambda = 1/2$ at $\alpha = \pi/4$, which we call \textsc{Midpoint}:
\[
P_{1/2}(a, b) = \frac{T_{\pi/4}(a, b) + S_{\pi/4}(a, b)}{2} = \left(\frac{b - a}{4},\, \frac{a + b + 2\max(a, b)}{4}\right).
\]
While $\beta^*$-\textsc{Hedge} achieves competitive ratio $5/4 = 1.25$ at $\alpha = \pi/4$, \textsc{Midpoint} strictly improves upon it. Moreover, its proof illustrates how lifting the analysis into the two-dimensional state space $\{(a, b) \in \mathbb{R}_{\ge 0}^2\}$ cleanly bypasses the geometric case explosion that arises in physical $(x, y)$ coordinates whenever the drone's trajectory crosses the vertical $y$-axis.

\begin{theorem}[Competitive Ratio of \textsc{Midpoint}]\label{thm:midpoint}
At scanning half-angle $\alpha = \pi/4$, the competitive ratio of \textsc{Midpoint} is tightly given by
\[
\rho(\textsc{Midpoint}; \pi/4) = \sqrt{\frac{3}{2}} \approx 1.2247 < \frac{5}{4}.
\]
\end{theorem}

\begin{proof}
Let $X = (x_0 = 0, x_1, \ldots, x_k)$ be any finite request sequence. For each round $i \in \{0, 1, \ldots, k\}$, define the dominant and non-dominant extreme magnitudes
\[
m_i := \max(a_i, b_i) \qquad\text{and}\qquad n_i := \min(a_i, b_i),
\]
so that $m_0 = n_0 = 0$, both sequences $\{m_i\}$ and $\{n_i\}$ are monotonically nondecreasing, and we write $(m, n) := (m_k, n_k)$ for the final bounds (assuming $m > 0$). We bound the total trajectory length $\textsc{Midpoint}(X) = \sum_{i=1}^k \|P_{1/2}(a_i, b_i) - P_{1/2}(a_{i-1}, b_{i-1})\|_2$ in four steps:

\paragraph{Step 1: Piecewise-affine structure across the state diagonal.}
The diagonal $a = b$ partitions the non-negative state quadrant $\mathbb{R}_{\ge 0}^2$ into two closed half-cones: $R_L := \{(a, b) : a \ge b\}$ (left-dominant, where $m = a$ and $n = b$) and $R_R := \{(a, b) : b \ge a\}$ (right-dominant, where $m = b$ and $n = a$). On $R_L$, the position rule
\[
P_{1/2}(m, n) = \left(\frac{n - m}{4},\, \frac{3m + n}{4}\right)
\]
is an exact affine function of $(m, n)$ with partial derivative vectors
\[
\frac{\partial P_{1/2}}{\partial m} = \left(-\frac{1}{4},\, \frac{3}{4}\right)
\qquad\text{and}\qquad
\frac{\partial P_{1/2}}{\partial n} = \left(\frac{1}{4},\, \frac{1}{4}\right),
\]
whose Euclidean norms are $\frac{\sqrt{10}}{4}$ and $\frac{\sqrt{2}}{4}$, respectively. On $R_R$, $P_{1/2}$ is the reflection of this affine map across the $y$-axis and has the exact same derivative norms with respect to $m = b$ and $n = a$. Along their shared boundary $a = b$, both formulas evaluate to $(0, a)$, so $P_{1/2}$ is continuous across the crease $a = b$.

\paragraph{Step 2: Splitting single-coordinate updates at $a = b$.}
At round $i$, the request $x_i$ updates at most one of the two coordinates $a$ or $b$:
\begin{itemize}
    \item If $x_i \in [-a_{i-1}, b_{i-1}]$, then $(a_i, b_i) = (a_{i-1}, b_{i-1})$ and the step displacement is zero.
    \item Suppose $x_i > b_{i-1}$, so $a_i = a_{i-1}$ and $b_i = x_i > b_{i-1}$ (the case $x_i < -a_{i-1}$ is symmetric). If $b_i \le a_i$ (the entire segment stays in $R_L$) or $b_{i-1} \ge a_i$ (the entire segment stays in $R_R$), the state transition $(a_{i-1}, b_{i-1}) \to (a_i, b_i)$ lies within a single affine region and changes either $n$ alone by $n_i - n_{i-1}$ or $m$ alone by $m_i - m_{i-1}$.
    \item If the update crosses the diagonal, i.e., $b_{i-1} < a_{i-1} = a_i < b_i$, insert the intermediate diagonal state $(a_i, a_i)$ and decompose the transition into two axis-aligned sub-segments:
    \[
    (a_i, b_{i-1}) \;\xrightarrow{\;\text{Leg 1 in } R_L\;}\; (a_i, a_i) \;\xrightarrow{\;\text{Leg 2 in } R_R\;}\; (a_i, b_i).
    \]
    Since $m_{i-1} = a_i$, $n_{i-1} = b_{i-1}$, $m_i = b_i$, and $n_i = a_i$, Leg~1 lies entirely in $R_L$ with $m$ held fixed at $a_i$ while $n$ increases from $n_{i-1}$ to $n_i$ (an increment of $n_i - n_{i-1}$), and Leg~2 lies entirely in $R_R$ with $n$ held fixed at $a_i$ while $m$ increases from $m_{i-1}$ to $m_i$ (an increment of $m_i - m_{i-1}$).
\end{itemize}

\begin{figure}[htbp]
\centering
\begin{tikzpicture}[
    scale=1.85,
    >=Stealth,
    target/.style={circle, fill=black, inner sep=1.3pt},
    optpt/.style={circle, fill=red!80!black, inner sep=1.5pt},
    bhpt/.style={circle, fill=green!55!black, inner sep=1.5pt},
    midpt/.style={circle, fill=blue!80!black, inner sep=1.5pt},
    sharedpt/.style={circle, fill=violet!80!black, inner sep=1.6pt},
    optline/.style={thick, dashed, red!80!black, ->},
    bhline/.style={thick, green!55!black, ->},
    midline/.style={thick, blue!80!black, ->},
    sharedline/.style={thick, violet!80!black, ->},
    coneline/.style={gray!55, dashed, thin}
]

  % =========================================================
  % PANEL (a): Physical Space (x, y) at alpha = pi/4
  % =========================================================
  \begin{scope}[shift={(0,0)}]
    \node[font=\bfseries\small] at (-0.35, 2.38) {(a) Physical Space $(x,y)$};

    % Axes
    \draw[->, thick] (-2.35, 0) -- (1.45, 0) node[right, font=\small] {$x$};
    \draw[->, thick] (0, 0) -- (0, 2.18) node[above, font=\small] {$y$};

    % Target Requests on x-axis
    \node[target, label=below:{\scriptsize $X_0(0)$}] (X0) at (0,0) {};
    \node[target, label=below:{\scriptsize $X_1(1)$}] (X1) at (1,0) {};
    \node[target, label=below:{\scriptsize $X_2(-2)$}] (X2) at (-2,0) {};

    % Feasibility Cone Walls (slope +-1 for alpha = pi/4)
    \draw[coneline] (-2,0) -- (0.12, 2.12);
    \draw[coneline] (1,0) -- (-1.12, 2.12);
    \draw[coneline] (0,0) -- (1.25, 1.25);

    % Key Coordinates
    \coordinate (P0)   at (0, 0);
    \coordinate (P1)   at (0.25, 0.75);
    \coordinate (Popt) at (-0.5, 1.5);
    \coordinate (Pmid) at (-0.25, 1.75);
    \coordinate (Pbh)  at (-0.125, 1.875);

    % Trajectories
    \draw[optline]    (P0) -- (Popt) node[midway, left=2pt, font=\scriptsize, text=red!80!black] {\textsc{Opt}};
    \draw[sharedline] (P0) -- (P1)   node[midway, right=2pt, font=\scriptsize, text=violet!80!black] {Step 1};
    \draw[midline]    (P1) -- (Pmid) node[midway, below left=-1pt, font=\scriptsize, text=blue!80!black] {\textsc{Midpoint}};
    \draw[bhline]     (P1) -- (Pbh)  node[midway, right=2pt, font=\scriptsize, text=green!55!black] {$\beta^*$-\textsc{Hedge}};

    % Point Markers & Labels
    \node[sharedpt, label=right:{\scriptsize $P_1(0.25, 0.75)$}] at (P1) {};
    \node[optpt, label=left:{\scriptsize $T_2(-0.5, 1.5)$}] at (Popt) {};
    \node[midpt, label=left:{\scriptsize $P_2^{\mathrm{Mid}}(-0.25, 1.75)$}] at (Pmid) {};
    \node[bhpt, label=above right:{\scriptsize $P_2^{\beta^*}(-0.125, 1.875)$}] at (Pbh) {};
  \end{scope}

  % =========================================================
  % PANEL (b): Coverage Space (u, v) at alpha = pi/4
  % =========================================================
  \begin{scope}[shift={(4.35,0)}]
    \node[font=\bfseries\small] at (1.25, 2.38) {(b) Coverage Space $(u,v)$};

    % Shaded Final Feasibility Quadrant K_{pi/4}(2, 1) = {u >= 2, v >= 1}
    \fill[gray!12] (2, 1) rectangle (2.45, 2.15);
    \draw[gray!50, dotted, thick] (2, 0) -- (2, 2.15) node[above, black, font=\scriptsize] {$u = 2$};
    \draw[gray!50, dotted, thick] (0, 1) -- (2.45, 1) node[right, black, font=\scriptsize] {$v = 1$};

    % Symmetry Diagonal u = v
    \draw[gray!35, densely dotted] (0, 0) -- (2.1, 2.1) node[above left, font=\tiny, text=gray!70] {$u=v$};

    % Axes
    \draw[->, thick] (0, 0) -- (2.55, 0) node[right, font=\small] {$u$};
    \draw[->, thick] (0, 0) -- (0, 2.18) node[above, font=\small] {$v$};

    % Key Coordinates
    \coordinate (U0)   at (0, 0);
    \coordinate (U1)   at (0.5, 1.0);
    \coordinate (Uopt) at (2.0, 1.0);
    \coordinate (Umid) at (2.0, 1.5);
    \coordinate (Ubh)  at (2.0, 1.75);

    % Trajectories
    \draw[optline]    (U0) -- (Uopt) node[midway, below=1pt, font=\scriptsize, text=red!80!black] {\textsc{Opt}};
    \draw[sharedline] (U0) -- (U1)   node[midway, left=1pt, font=\scriptsize, text=violet!80!black] {Step 1};
    \draw[midline]    (U1) -- (Umid) node[midway, below=1pt, font=\scriptsize, text=blue!80!black] {slope $1/3$};
    \draw[bhline]     (U1) -- (Ubh)  node[midway, above left=-1pt, font=\scriptsize, text=green!55!black] {slope $1/2$};

    % Point Markers & Labels
    \node[target, label=below left:{\scriptsize $(0,0)$}] at (U0) {};
    \node[sharedpt, label=above left:{\scriptsize $(0.5, 1)$}] at (U1) {};
    \node[optpt, label=below right:{\scriptsize $T_2(2, 1)$}] at (Uopt) {};
    \node[midpt, label=right:{\scriptsize $(2, 1.5)$}] at (Umid) {};
    \node[bhpt, label=right:{\scriptsize $(2, 1.75)$}] at (Ubh) {};
  \end{scope}

\end{tikzpicture}
\caption{Comparison of $\beta^*$-\textsc{Hedge} ($\beta^*(\pi/4) = \arctan(1/3)$, $\sigma = 1/2$) and \textsc{Midpoint} ($\mathcal{A}_{1/2}$) on $X = (0, 1, -2)$ at $\alpha = \pi/4$. Both algorithms share the first step to $(u_1, v_1) = (0.5, 1)$. Upon receiving $x_2 = -2$, $\beta^*$-\textsc{Hedge} maintains fixed slope $\sigma = 1/2$ in coverage space and overshoots to $(2, 1.75)$ (ratio $5/4 = 1.25$), whereas \textsc{Midpoint} sheds excess right-side cushion to reach $(2, 1.5)$ at lower altitude (ratio $\frac{1+\sqrt{2}}{2} \approx 1.2071$).}
\label{fig:beta-vs-midpoint}
\end{figure}

\paragraph{Step 3: Single-step triangle-inequality bound.}
When the update crosses $a = b$, the triangle inequality in physical space bounds the drone's direct straight-line flight by the sum of the two sub-segment lengths:
\[
\|P_{1/2}(a_i, b_i) - P_{1/2}(a_{i-1}, b_{i-1})\|_2
\le \|P_{1/2}(a_i, a_i) - P_{1/2}(a_i, b_{i-1})\|_2 + \|P_{1/2}(a_i, b_i) - P_{1/2}(a_i, a_i)\|_2.
\]
Because $P_{1/2}$ is affine on $R_L$ and $R_R$ with unit rates $\frac{\sqrt{10}}{4}$ along $m$ and $\frac{\sqrt{2}}{4}$ along $n$, the two terms on the right equal $\frac{\sqrt{2}}{4}(n_i - n_{i-1})$ and $\frac{\sqrt{10}}{4}(m_i - m_{i-1})$, respectively. The exact same upper bound holds trivially when the update does not cross $a = b$. Thus, for every round $i \in \{1, \ldots, k\}$,
\begin{equation}\label{eq:midpoint-step-bound}
\|P_{1/2}(a_i, b_i) - P_{1/2}(a_{i-1}, b_{i-1})\|_2 \le \frac{\sqrt{10}}{4}(m_i - m_{i-1}) + \frac{\sqrt{2}}{4}(n_i - n_{i-1}).
\end{equation}

\paragraph{Step 4: Telescoping and ratio maximization.}
Summing~\eqref{eq:midpoint-step-bound} over $i = 1, \ldots, k$ cancels all intermediate terms and uses $m_0 = n_0 = 0$:
\[
\textsc{Midpoint}(X)
\le \sum_{i=1}^k \left[\frac{\sqrt{10}}{4}(m_i - m_{i-1}) + \frac{\sqrt{2}}{4}(n_i - n_{i-1})\right]
= \frac{\sqrt{10}}{4}m + \frac{\sqrt{2}}{4}n.
\]
At $\alpha = \pi/4$ ($c = 0$), Theorem~\ref{thm:opt} gives $\mathrm{OPT}(X) = \frac{1}{\sqrt{2}}\sqrt{m^2 + n^2}$. Applying the Cauchy--Schwarz inequality yields
\[
\frac{\textsc{Midpoint}(X)}{\mathrm{OPT}(X)} \le \frac{\sqrt{5}\,m + n}{2\sqrt{m^2 + n^2}} \le \frac{\sqrt{5 + 1}}{2} = \sqrt{\frac{3}{2}}.
\]
For tightness, set $r = 1/\sqrt{5}$ and consider the three-request sequence $X = (0, -r, r, -1)$. The cumulative extreme pairs $(a_i, b_i)$ are $(r, 0)$, $(r, r)$, and $(1, r)$: they all lie in $R_L$ so \textsc{Midpoint} visits $\big(-\frac{r}{4}, \frac{3r}{4}\big)$, $(0, r)$, and $\big(\frac{r - 1}{4}, \frac{3 + r}{4}\big)$ with step lengths $\frac{\sqrt{10}}{4}r$, $\frac{\sqrt{2}}{4}r$, and $\frac{\sqrt{10}}{4}(1 - r)$, incurring total cost $\frac{\sqrt{10} + \sqrt{2}\,r}{4}$. Dividing by $\mathrm{OPT}(X) = \frac{\sqrt{1 + r^2}}{\sqrt{2}}$ gives $\frac{\sqrt{5} + 1/\sqrt{5}}{2\sqrt{1 + 1/5}} = \sqrt{\frac{3}{2}}$.
\end{proof}

\subsection{The General State-Based Family $\mathcal{A}_\lambda$}\label{subsec:a-lambda}

We now extend the piecewise-affine interpolation and telescoping argument of Theorem~\ref{thm:midpoint} to arbitrary $\alpha \in (0, \pi/2)$ and $\lambda \in [0, 1]$.

\begin{lemma}[Linear Cost Bound for $\mathcal{A}_\lambda$]\label{lem:a-lambda-cost}
For every $\alpha \in (0, \pi/2)$ and $\lambda \in [0, 1]$, define the dominant and non-dominant unit cost multipliers
\begin{equation}\label{eq:mu-nu-def}
\mu_\lambda = \frac{1}{2}\sqrt{\lambda^2 + (2 - \lambda)^2\cot^2\alpha}
\qquad\text{and}\qquad
\nu_\lambda = \frac{\lambda}{2\sin\alpha}.
\end{equation}
Then for every finite request sequence $X = (x_0 = 0, x_1, \ldots, x_k)$ with final bounds $m = \max(a_k, b_k)$ and $n = \min(a_k, b_k)$,
\begin{equation}\label{eq:a-lambda-cost-ub}
\mathcal{A}_\lambda(X) \le \mu_\lambda m + \nu_\lambda n.
\end{equation}
Moreover, for every $r \in [0, 1]$, equality holds in~\eqref{eq:a-lambda-cost-ub} on the three-request sequence $X = (0, -r, r, -1)$ (with $m = 1$ and $n = r$).
\end{lemma}

\begin{proof}
As in Theorem~\ref{thm:midpoint}, let $m_i := \max(a_i, b_i)$ and $n_i := \min(a_i, b_i)$ for each $i \in \{0, 1, \ldots, k\}$. On the closed half-cone $R_L = \{a \ge b\}$ (where $m = a$ and $n = b$),~\eqref{eq:p-lambda-def} evaluates to the affine map
\[
P_\lambda(m, n) = \left(\frac{\lambda(n - m)}{2},\, \Big(\Big(1 - \frac{\lambda}{2}\Big)m + \frac{\lambda}{2}n\Big)\cot\alpha\right),
\]
whose partial derivative vectors with respect to $m$ and $n$ are
\[
\frac{\partial P_\lambda}{\partial m} = \left(-\frac{\lambda}{2},\, \Big(1 - \frac{\lambda}{2}\Big)\cot\alpha\right)
\qquad\text{and}\qquad
\frac{\partial P_\lambda}{\partial n} = \left(\frac{\lambda}{2},\, \frac{\lambda}{2}\cot\alpha\right),
\]
with Euclidean norms $\mu_\lambda$ and $\nu_\lambda$, respectively. By reflection across the vertical $y$-axis, the same derivative norms hold on $R_R = \{b \ge a\}$, and both expressions agree at $(0, a\cot\alpha)$ along the shared crease $a = b$.

Following Steps~2--3 of Theorem~\ref{thm:midpoint}, any single-coordinate update $(a_{i-1}, b_{i-1}) \to (a_i, b_i)$ that crosses the diagonal $a = b$ splits at the intermediate diagonal state into a pure $(n_i - n_{i-1})$-increasing leg on one half-cone and a pure $(m_i - m_{i-1})$-increasing leg on the other. By the triangle inequality and linearity on each half-cone, every round $i \in \{1, \ldots, k\}$ satisfies
\[
\|P_\lambda(a_i, b_i) - P_\lambda(a_{i-1}, b_{i-1})\|_2 \le \mu_\lambda(m_i - m_{i-1}) + \nu_\lambda(n_i - n_{i-1}).
\]
Since $m_0 = n_0 = 0$ and $(m_k, n_k) = (m, n)$, summing over $i = 1, \ldots, k$ telescopes directly to~\eqref{eq:a-lambda-cost-ub}.

On $X = (0, -r, r, -1)$, the successive extreme pairs $(a_i, b_i)$ are $(r, 0)$, $(r, r)$, and $(1, r)$, which all lie in $R_L$. The three consecutive displacements are purely $m$-increasing by $r$, purely $n$-increasing by $r$, and purely $m$-increasing by $1 - r$, so their lengths sum to $\mu_\lambda r + \nu_\lambda r + \mu_\lambda(1 - r) = \mu_\lambda + \nu_\lambda r$.
\end{proof}

\begin{theorem}[Competitive Ratio of $\mathcal{A}_\lambda$]\label{thm:a-lambda-cr}
For any scanning half-angle $\alpha \in (0, \pi/2)$ and any $\lambda \in [0, 1]$, let $c = \cos(2\alpha)$,
\[
L = 2\sin\alpha\,\mu_\lambda = \sqrt{\lambda^2 + 2(1 + c)(1 - \lambda)},
\qquad\text{and}\qquad
r_\lambda = \max\!\left(0,\, \frac{\lambda - cL}{L - c\lambda}\right) \in [0, 1].
\]
Then the competitive ratio of $\mathcal{A}_\lambda$ is tightly given by
\begin{equation}\label{eq:cr-a-lambda}
\rho(\mathcal{A}_\lambda; \alpha) = \frac{L + \lambda r_\lambda}{\sqrt{1 + r_\lambda^2 + 2c r_\lambda}}.
\end{equation}
In particular, whenever $\lambda \ge cL$ (which holds for all $\lambda \in [0, 1]$ when $\alpha \ge \pi/4$), this simplifies to
\begin{equation}\label{eq:cr-a-lambda-interior}
\rho(\mathcal{A}_\lambda; \alpha) = \sqrt{\frac{L^2 - 2c\lambda L + \lambda^2}{1 - c^2}}.
\end{equation}
\end{theorem}

\begin{proof}
First observe that $L^2 - \lambda^2 = 2(1 + c)(1 - \lambda) = 4\cos^2\alpha(1 - \lambda) \ge 0$, so $L \ge \lambda \ge 0$ (with $L > 0$ because $\lambda$ and $1 - \lambda$ cannot both vanish). Since $|c| < 1$, we also have $L - c\lambda > 0$.

By Lemma~\ref{lem:a-lambda-cost}, $\mathcal{A}_\lambda(X) \le \frac{m}{2\sin\alpha}(L + \lambda r)$ for every sequence $X$ with $m > 0$ and $r = n/m \in [0, 1]$, with equality on $X = (0, -r, r, -1)$. Comparing with Theorem~\ref{thm:opt} across the two regimes of $r$:
\begin{itemize}
    \item \emph{Apex regime ($r \in [\max(0, -c), 1]$).} Here $\mathrm{OPT}(X) = \frac{m}{2\sin\alpha}\sqrt{1 + r^2 + 2cr}$, so the ratio is bounded by $h(r) := \frac{L + \lambda r}{\sqrt{1 + r^2 + 2cr}}$. Differentiating gives
    \[
    h'(r) = \frac{(\lambda - cL) - r(L - c\lambda)}{(1 + r^2 + 2cr)^{3/2}}.
    \]
    Setting $\tilde{r} = \frac{\lambda - cL}{L - c\lambda}$, we have $1 - \tilde{r} = \frac{(1 + c)(L - \lambda)}{L - c\lambda} \ge 0$ (since $L \ge \lambda$), so $\tilde{r} \le 1$ always holds. Moreover, when $c < 0$, $\tilde{r} - (-c) = \frac{\lambda(1 - c^2)}{L - c\lambda} \ge 0$, so $\tilde{r} \ge -c$. Hence $h(r)$ is maximized on $[\max(0, -c), 1]$ at $r_\lambda = \max(0, \tilde{r})$.
    \item \emph{Off-apex regime ($c < 0$ and $r \in [0, -c)$).} Here $\mathrm{OPT}(X) = m\cos\alpha$, so the ratio is bounded by $\frac{L + \lambda r}{\sin(2\alpha)} \le \frac{L - c\lambda}{\sin(2\alpha)} = h(-c) \le h(r_\lambda)$, as $\lambda \ge 0$ and $\sqrt{1 - c^2} = \sin(2\alpha)$.
\end{itemize}
Tightness follows because $X = (0, -r_\lambda, r_\lambda, -1)$ attains $\mathcal{A}_\lambda(X) = \frac{L + \lambda r_\lambda}{2\sin\alpha}$ by Lemma~\ref{lem:a-lambda-cost}, while $r_\lambda \ge \max(0, -c)$ places $\mathrm{OPT}(X)$ at the cone apex by Theorem~\ref{thm:opt}. Finally, when $\lambda \ge cL$, substituting $r_\lambda = \frac{\lambda - cL}{L - c\lambda}$ into~\eqref{eq:cr-a-lambda} gives
\[
L + \lambda r_\lambda = \frac{L^2 - 2c\lambda L + \lambda^2}{L - c\lambda}
\qquad\text{and}\qquad
1 + r_\lambda^2 + 2c r_\lambda = \frac{(1 - c^2)(L^2 - 2c\lambda L + \lambda^2)}{(L - c\lambda)^2},
\]
yielding~\eqref{eq:cr-a-lambda-interior}.
\end{proof}

\subsection{Comparison with $\beta^*$-\textsc{Hedge}}\label{subsec:a-lambda-vs-beta}

Optimizing $\lambda \in [0, 1]$ for each $\alpha \in (0, \pi/2)$ yields a clear structural comparison with fixed-angle $\beta^*$-\textsc{Hedge} (Theorem~\ref{thm:opt-beta}): state-based interpolation never performs worse than $\beta^*$-\textsc{Hedge}, strictly outperforms it across the entire interval $\alpha \in (0, \pi/3)$, and coincides with it on $[\pi/3, \pi/2)$.

\begin{theorem}\label{thm:a-lambda-vs-beta}
For every scanning half-angle $\alpha \in (0, \pi/2)$, let $\rho^*(\mathcal{A}; \alpha) := \min_{\lambda \in [0, 1]} \rho(\mathcal{A}_\lambda; \alpha)$. Then:
\begin{enumerate}
    \item For every $\alpha \in (0, \pi/3)$, $\rho^*(\mathcal{A}; \alpha) < \rho(\beta^*\text{-}\textsc{Hedge}; \alpha)$.
    \item For every $\alpha \in [\pi/3, \pi/2)$, $\rho^*(\mathcal{A}; \alpha) = \rho(\beta^*\text{-}\textsc{Hedge}; \alpha) = \csc\alpha$, uniquely attained at $\lambda^* = 0$ (\textsc{Straight-Up}).
\end{enumerate}
\end{theorem}

\begin{proof}
We examine the three angular regimes of Theorem~\ref{thm:opt-beta}:
\begin{itemize}
    \item \emph{Acute angles $\alpha \in (0, \pi/6]$.} Here $\rho(\beta^*\text{-}\textsc{Hedge}; \alpha) = \sec\alpha = \rho(\mathcal{A}_1; \alpha)$. Since $\lambda = 1 > c$ for every $\alpha \in (0, \pi/2)$,~\eqref{eq:cr-a-lambda-interior} holds on a neighborhood of $\lambda = 1$. Differentiating $\rho(\mathcal{A}_\lambda; \alpha)^2$ using $L^2 = \lambda^2 + 2(1 + c)(1 - \lambda)$ and $L\big|_{\lambda=1} = 1$ gives
    \[
    \left.\frac{\partial}{\partial \lambda}\,\rho(\mathcal{A}_\lambda; \alpha)^2\right|_{\lambda = 1}
    = \frac{2\big(\lambda - (1 + c)\big) - 2c\big(L + \lambda\frac{\lambda - 1 - c}{L}\big) + 2\lambda}{1 - c^2}\Bigg|_{\lambda = 1}
    = \frac{2(1 - c)}{1 + c} = 2\tan^2\alpha > 0.
    \]
    Hence $\rho(\mathcal{A}_{1 - \varepsilon}; \alpha) < \sec\alpha$ for all sufficiently small $\varepsilon > 0$.

    \item \emph{Intermediate angles $\alpha \in (\pi/6, \pi/3)$.} Let $\beta^* = \beta^*(\alpha) \in (0, \alpha)$ be the optimal steering angle from Theorem~\ref{thm:opt-beta}, and let $\lambda_0 = \frac{2\tan\beta^*}{\tan\alpha + \tan\beta^*} \in (0, 1)$ be its corresponding hedging weight. Comparing~\eqref{eq:alg-ub-beta} and~\eqref{eq:a-lambda-cost-ub}, we have
    \[
    \beta^*\text{-}\textsc{Hedge}(X) \le \mu_{\lambda_0}(m + \lambda_0 n)
    \qquad\text{and}\qquad
    \mathcal{A}_{\lambda_0}(X) \le \mu_{\lambda_0} m + \frac{\lambda_0}{2\sin\alpha} n,
    \]
    because $\frac{\cos\alpha}{\sin(\alpha + \beta^*)} = N_\alpha(1, 1 - \lambda_0) = \mu_{\lambda_0}$. Moreover,
    \[
    (2\sin\alpha\,\mu_{\lambda_0})^2 - 1 = L^2 - 1 = (1 - \lambda_0)(1 + 2c - \lambda_0) > 0,
    \]
    so $\mu_{\lambda_0} > \frac{1}{2\sin\alpha}$. Since the worst-case ratio $r_{\lambda_0}$ of $\mathcal{A}_{\lambda_0}$ satisfies $r_{\lambda_0} \ge \frac{\lambda_0 - cL}{L - c\lambda_0} > 0$, it follows that
    \[
    \rho(\mathcal{A}_{\lambda_0}; \alpha) = \frac{\mu_{\lambda_0} + \frac{\lambda_0}{2\sin\alpha}\,r_{\lambda_0}}{\mathrm{OPT}(1, r_{\lambda_0})}
    < \frac{\mu_{\lambda_0}(1 + \lambda_0 r_{\lambda_0})}{\mathrm{OPT}(1, r_{\lambda_0})}
    \le \rho(\beta^*\text{-}\textsc{Hedge}; \alpha).
    \]

    \item \emph{Wide angles $\alpha \in [\pi/3, \pi/2)$.} Here $c = \cos(2\alpha) \le -1/2 < 0$, so~\eqref{eq:cr-a-lambda-interior} applies for all $\lambda \in [0, 1]$. Using $\csc^2\alpha = \frac{2}{1 - c} = \frac{2(1 + c)}{1 - c^2}$ and $L^2 = \lambda^2 + 2(1 + c)(1 - \lambda)$, we obtain the exact identity
    \[
    \rho(\mathcal{A}_\lambda; \alpha)^2 - \csc^2\alpha
    = \frac{2\lambda\big(\lambda - 1 - c - cL\big)}{1 - c^2}.
    \]
    For any $\lambda \in (0, 1]$: if $\lambda \ge 1 + c$, then $\lambda - 1 - c - cL \ge -cL > 0$; if $0 < \lambda < 1 + c$, then $1 + c - \lambda > 0$ while
    \[
    (-cL)^2 - (1 + c - \lambda)^2 = (1 - \lambda)(1 - c^2)(1 + 2c - \lambda) > 0
    \]
    (as $1 + 2c \le 0 < \lambda$ and $\lambda < 1 + c \le 1/2 < 1$), so $-cL > 1 + c - \lambda$. Thus $\rho(\mathcal{A}_\lambda; \alpha) > \csc\alpha$ for every $\lambda \in (0, 1]$, and the minimum $\csc\alpha$ is uniquely achieved at $\lambda^* = 0$.
\end{itemize}
\end{proof}
\section{Universal Lower Bound}\label{sec:lower-bound}

We now establish a universal lower bound on the competitive ratio of any online algorithm across the entire non-degenerate angle domain $\alpha \in (0, \pi/2)$. Analyzing finite prefixes of an alternating geometric request sequence by averaging normalized active and passive coverage reaches across rounds yields a scale-parameterized lower bound without assuming that the drone hugs a cone boundary or stays on one side of the vertical axis. Optimizing the geometric growth scale $s > 1$ characterizes the universal deterministic lower bound via a single quartic root, and a convexity-based derandomization argument at the end of the section extends the bound, and the resulting competitive ratio gap at $\alpha = \pi/4$, to all randomized online algorithms against an oblivious adversary.

\subsection{Geometric Lower Bound for a Fixed Growth Scale}\label{subsec:geometric-lb}

Fix a scanning half-angle $\alpha \in (0, \pi/2)$ and a geometric growth factor $s > 1$. Consider the alternating sequence of ground targets
\[
x_0 = 0, \qquad x_i = (-1)^{i-1} s^{i-1} \quad (i = 1, 2, \ldots),
\]
so $x_1 = 1$, $x_2 = -s$, $x_3 = s^2$, and in general $x_i$ lands on the right ($+s^{i-1}$) when $i$ is odd and on the left ($-s^{i-1}$) when $i$ is even. All competitive comparisons are made strictly on finite prefixes $X^{(i)} = (x_0, x_1, \ldots, x_i)$.

\begin{theorem}[Scale-Parameterized Lower Bound]\label{thm:lb-fixed-s}
For every scanning half-angle $\alpha \in (0, \pi/2)$ and every growth factor $s > 1$, the competitive ratio of every deterministic online algorithm is at least
\begin{equation}\label{eq:f-alpha-s}
F_\alpha(s) := \frac{s(s + 1)\sin(2\alpha)}{\sqrt{1 + s^4 - 2s^2\cos(4\alpha)}}.
\end{equation}
\end{theorem}

\begin{proof}
Suppose a deterministic online algorithm has finite competitive ratio $\rho$, and let $(u_i, v_i)$ and $C_i$ denote its coverage reaches and cumulative Euclidean path length after request $x_i$ (with $(u_0, v_0) = (0, 0)$ and $C_0 = 0$). The proof proceeds in three steps: we first normalize each round's active and passive reaches by $s^{i-1}$; we then average the normalized step vectors over rounds $2, \ldots, M$ and let $M \to \infty$ so that any drift telescopes away onto a single stationary pair $(p, q)$; and finally we minimize the single-step cost over $(p, q)$ by completing the square.

\paragraph{Step 1: Normalized active and passive reaches.}
After round $i \ge 1$, define the \emph{normalized active reach} $p_i$ (on the side just requested) and the \emph{normalized passive reach} $q_i$ (on the opposite side) by scaling by $s^{i-1}$:
\[
(p_i, q_i) :=
\begin{cases}
(v_i / s^{i-1},\, u_i / s^{i-1}) & \text{if } i \text{ is odd (right request } x_i = +s^{i-1}\text{)},\\[3pt]
(u_i / s^{i-1},\, v_i / s^{i-1}) & \text{if } i \text{ is even (left request } x_i = -s^{i-1}\text{)}.
\end{cases}
\]
For every $i \ge 2$, covering $x_i$ and the previous extreme $x_{i-1}$ requires the active reach to be at least $s^{i-1}$ and the passive reach to be at least $s^{i-2}$, so
\begin{equation}\label{eq:pq-feas}
p_i \ge 1 \qquad\text{and}\qquad q_i \ge \frac{1}{s}.
\end{equation}
Notice that we impose no ordering between $p_i$ and $q_i$, so the drone is free to hover strictly inside the feasibility cone or cross the vertical $y$-axis on any round.

To bound the offline cost on $X^{(i)}$ ($i \ge 2$), observe that flying straight from the origin to the stationary position with coverage reaches $(s^{i-1}, s^{i-2})$ (for odd $i$) or $(s^{i-2}, s^{i-1})$ (for even $i$) covers the entire prefix $X^{(i)}$ at cost $s^{i-1} D$, where $D := N_\alpha(1, 1/s)$. Hence $\mathrm{OPT}(X^{(i)}) \le s^{i-1} D$, and $\rho$-competitiveness gives $C_i \le \rho s^{i-1} D$. Because $N_\alpha$ is a norm on $\mathbb{R}^2$ and $(u_0, v_0) = (0, 0)$, the triangle inequality implies $N_\alpha(u_i, v_i) \le C_i \le \rho s^{i-1} D$; dividing by $s^{i-1}$ shows that every $(p_i, q_i)$ lies in the compact set $\mathcal{K} := \{(p, q) \in \mathbb{R}_{\ge 0}^2 : N_\alpha(p, q) \le \rho D\}$.

Now examine the step length $\Delta C_i = C_i - C_{i-1}$ in round $i \ge 2$. Because consecutive requests alternate sides, round $i-1$'s active reach $p_{i-1}$ sits on round $i$'s passive side (scaled down by $1/s$), while round $i-1$'s passive reach $q_{i-1}$ sits on round $i$'s active side (scaled down by $1/s$). By the symmetry $N_\alpha(\Delta u, \Delta v) = N_\alpha(\Delta v, \Delta u)$, both even and odd rounds give
\[
\Delta C_i = s^{i-1}\, N_\alpha\!\left(p_i - \frac{q_{i-1}}{s},\; q_i - \frac{p_{i-1}}{s}\right).
\]
Setting $t_i := C_i / s^{i-1} \in [0, \rho D]$ and dividing $C_i = C_{i-1} + \Delta C_i$ by $s^{i-1}$ yields
\begin{equation}\label{eq:t-recurrence}
N_\alpha\!\left(p_i - \frac{q_{i-1}}{s},\; q_i - \frac{p_{i-1}}{s}\right) = t_i - \frac{t_{i-1}}{s}
\qquad (i \ge 2).
\end{equation}

\paragraph{Step 2: Telescoping and averaging across rounds.}
Summing~\eqref{eq:t-recurrence} over $i = 2, 3, \ldots, M$ and using $0 \le t_i \le \rho D$ gives
\[
\sum_{i=2}^M N_\alpha\!\left(p_i - \frac{q_{i-1}}{s},\; q_i - \frac{p_{i-1}}{s}\right)
= \left(1 - \frac{1}{s}\right)\sum_{i=2}^{M-1} t_i + t_M - \frac{t_1}{s}
\le \left((M - 1)\left(1 - \frac{1}{s}\right) + \frac{1}{s}\right)\rho D.
\]
Let $(\bar{p}_M, \bar{q}_M) := \frac{1}{M - 1}\sum_{i=2}^M (p_i, q_i)$. Dividing by $M - 1$ and applying the triangle inequality for $N_\alpha$ yields
\[
N_\alpha\!\left(\bar{p}_M - \frac{\bar{q}_M}{s} + \frac{q_M - q_1}{s(M - 1)},\; \bar{q}_M - \frac{\bar{p}_M}{s} + \frac{p_M - p_1}{s(M - 1)}\right)
\le \left(1 - \frac{1}{s} + \frac{1}{s(M - 1)}\right)\rho D.
\]
By convexity of $\mathcal{K}$, every $(\bar{p}_M, \bar{q}_M)$ lies in the compact set $\mathcal{K} \cap ([1, \infty) \times [1/s, \infty))$, while the boundary terms $\frac{q_M - q_1}{s(M-1)}$ and $\frac{p_M - p_1}{s(M-1)}$ vanish as $M \to \infty$. Passing to a convergent subsequence therefore produces a single stationary pair $(p, q)$ with $p \ge 1$ and $q \ge 1/s$ satisfying
\begin{equation}\label{eq:stationary-ineq}
N_\alpha\!\left(p - \frac{q}{s},\; q - \frac{p}{s}\right) \le \left(1 - \frac{1}{s}\right)\rho D.
\end{equation}

\paragraph{Step 3: Two-square minimization.}
Let $A = p - q/s$ and $B = q - p/s$, so that $A + B = \frac{s - 1}{s}(p + q)$ and $A - B = \frac{s + 1}{s}(p - q)$. Using the trigonometric identity
\[
A^2 + B^2 + 2\cos(2\alpha)\,AB = \cos^2\alpha\,(A + B)^2 + \sin^2\alpha\,(A - B)^2,
\]
we can rewrite $Q(p, q) := 4\sin^2\alpha\, N_\alpha(A, B)^2$ as a sum of two squares:
\[
Q(p, q) = W_+(q + p)^2 + W_-(q - p)^2,
\qquad
W_+ := \left(\frac{s - 1}{s}\right)^2\cos^2\alpha, \quad
W_- := \left(\frac{s + 1}{s}\right)^2\sin^2\alpha.
\]
For any $W_+, W_- > 0$ and any $p, q \in \mathbb{R}$, completing the square in $q$ gives
\[
W_+(q + p)^2 + W_-(q - p)^2
= (W_+ + W_-)\!\left(q + \frac{W_+ - W_-}{W_+ + W_-}\,p\right)^2 + \frac{4W_+ W_-}{W_+ + W_-}\,p^2
\ge \frac{4W_+ W_-}{W_+ + W_-}\,p^2.
\]
Relaxing $q \ge 1/s$ to $q \in \mathbb{R}$ and using $p \ge 1$, we evaluate
\[
4W_+ W_- = \frac{(s - 1)^2(s + 1)^2\sin^2(2\alpha)}{s^4}
\qquad\text{and}\qquad
W_+ + W_- = \frac{s^2 - 2\cos(2\alpha)\,s + 1}{s^2},
\]
which yields
\[
Q(p, q) \ge \frac{(s - 1)^2(s + 1)^2\sin^2(2\alpha)}{s^2\big(s^2 - 2\cos(2\alpha)\,s + 1\big)}.
\]
On the other hand, squaring~\eqref{eq:stationary-ineq} and multiplying by $4\sin^2\alpha$ gives
\[
Q(p, q) \le \left(\frac{s - 1}{s}\right)^2 \rho^2 \cdot \frac{s^2 + 2\cos(2\alpha)\,s + 1}{s^2}.
\]
Combining the two bounds on $Q(p, q)$ and canceling $(s - 1)^2 / s^2 > 0$, we obtain
\[
\frac{(s + 1)^2\sin^2(2\alpha)}{s^2 - 2\cos(2\alpha)\,s + 1}
\le \frac{s^2 + 2\cos(2\alpha)\,s + 1}{s^2}\,\rho^2.
\]
Since $(s^2 - 2\cos(2\alpha)\,s + 1)(s^2 + 2\cos(2\alpha)\,s + 1) = (s^2 + 1)^2 - 4\cos^2(2\alpha)\,s^2 = 1 + s^4 - 2s^2\cos(4\alpha)$, multiplying through by $s^2(s^2 - 2\cos(2\alpha)\,s + 1)$ and taking square roots yields $\rho \ge F_\alpha(s)$.
\end{proof}

\subsection{Optimizing the Scale Across All Angles}\label{subsec:opt-scale}

Maximizing $F_\alpha(s)$ over $s > 1$ characterizes the universal deterministic lower bound $\rho^*_{\mathrm{det}}(\alpha)$ at every non-degenerate angle $\alpha \in (0, \pi/2)$ via a single quartic root.

\begin{theorem}[Universal Deterministic Lower Bound]\label{thm:lb-universal}
For every scanning half-angle $\alpha \in (0, \pi/2)$, the quartic polynomial
\begin{equation}\label{eq:quartic-p}
P_\alpha(s) := s^4 + 2\cos(4\alpha)\,s^3 - 2s - 1
\end{equation}
has a unique root $s^* = s^*(\alpha) > 1$, and every deterministic online algorithm has competitive ratio at least
\begin{equation}\label{eq:lb-star}
\rho^*_{\mathrm{det}}(\alpha) \ge F_\alpha(s^*)
= \frac{s^*(s^* + 1)\sin(2\alpha)}{\sqrt{1 + (s^*)^4 - 2(s^*)^2\cos(4\alpha)}}.
\end{equation}
Moreover:
\begin{enumerate}
    \item The lower bound $F_\alpha(s^*)$ is invariant under the reflection $\alpha \mapsto \pi/2 - \alpha$.
    \item Whenever $\cos(2\alpha) \ne 0$, the maximizing scale satisfies $1 < s^* < 1/|\cos(2\alpha)|$; consequently, on every prefix $X^{(i)}$ ($i \ge 2$) of the maximizing geometric sequence, the offline optimum $\mathrm{OPT}(X^{(i)})$ is attained at the cone apex.
\end{enumerate}
\end{theorem}

\begin{proof}
Differentiating~\eqref{eq:f-alpha-s} with respect to $s$ gives
\[
F'_\alpha(s) = \frac{-\sin(2\alpha)\,P_\alpha(s)}{\big(1 + s^4 - 2s^2\cos(4\alpha)\big)^{3/2}}.
\]
At $s = 1$, $P_\alpha(1) = 2(\cos(4\alpha) - 1) = -4\sin^2(2\alpha) < 0$ because $\alpha \in (0, \pi/2)$, whereas $\lim_{s \to \infty} P_\alpha(s) = +\infty$. Furthermore, $P''_\alpha(s) = 12s\big(s + \cos(4\alpha)\big) > 0$ for all $s > 1$ (since $\cos(4\alpha) \ge -1$), so $P_\alpha$ is strictly convex on $[1, \infty)$ and therefore possesses a unique root $s^* > 1$. Since $F'_\alpha(s) > 0$ on $(1, s^*)$ and $F'_\alpha(s) < 0$ on $(s^*, \infty)$, $F_\alpha(s)$ attains its unique global maximum on $(1, \infty)$ at $s^*$, establishing~\eqref{eq:lb-star} via Theorem~\ref{thm:lb-fixed-s}.

Under $\alpha \mapsto \pi/2 - \alpha$, both $\sin(2\alpha)$ and $\cos(4\alpha)$ are unchanged, leaving $F_\alpha$, $P_\alpha$, and $s^*$ identical. Finally, when $\cos(2\alpha) \ne 0$, set $t := |\cos(2\alpha)| \in (0, 1)$, so $\cos(4\alpha) = 2t^2 - 1$. Evaluating $P_\alpha$ at $1/t$ yields
\[
t^4 P_\alpha(1/t) = 1 + 2(2t^2 - 1)t - 2t^3 - t^4 = 1 - 2t + 2t^3 - t^4 = (1 - t)^3(1 + t) > 0.
\]
Since $P_\alpha(1) < 0 < P_\alpha(1/t)$ and $s^*$ is the unique root in $(1, \infty)$, it follows that $1 < s^* < 1/t$. On every prefix $X^{(i)}$ ($i \ge 2$), the ratio of the smaller to larger extreme bound is $r = 1/s^* > |\cos(2\alpha)| \ge -\cos(2\alpha)$, which is precisely the condition in Theorem~\ref{thm:opt} placing $\mathrm{OPT}(X^{(i)})$ at the cone apex.
\end{proof}

\subsection{Randomization and the Competitive Ratio Gap at $\alpha = \pi/4$ }

Finally, convexity of both the feasibility cone and the Euclidean distance metric implies that randomization against an oblivious adversary cannot improve upon the best deterministic competitive ratio. Let $\rho^*_{\mathrm{rand}}(\alpha)$ denote the infimum expected competitive ratio over all almost-surely feasible randomized online algorithms at scanning half-angle $\alpha$.

\begin{lemma}[Derandomization on Fixed Sequences]\label{lem:derandomization}
Let $\mathcal{R}$ be any randomized online algorithm that is almost surely feasible and incurs finite expected movement cost on every finite request sequence. Then for every fixed finite request sequence $X = (x_0, x_1, \ldots, x_k)$, there exists a feasible deterministic trajectory on $X$ whose total Euclidean length is at most $\mathbb{E}[\mathcal{R}(X)]$. In particular,
\[
\rho^*_{\mathrm{rand}}(\alpha) = \rho^*_{\mathrm{det}}(\alpha) \ge F_\alpha(s^*)
\qquad\text{for all }\alpha \in (0, \pi/2).
\]
\end{lemma}

\begin{proof}
For each round $i \in \{0, 1, \ldots, k\}$, let $Z_i = (X^{\mathrm{alg}}_i, Y^{\mathrm{alg}}_i)$ denote the random physical position of $\mathcal{R}$ after processing the prefix $(x_0, \ldots, x_i)$, with $Z_0 = (0, 0)$, and let $a_i, b_i$ be the cumulative left and right request extremes of $(x_0, \ldots, x_i)$. Because $\|Z_i\|_2 \le \sum_{j=1}^i \|Z_j - Z_{j-1}\|_2$ and the right-hand side has finite expectation, the expected position $\bar{Z}_i := \mathbb{E}[Z_i]$ is well-defined.

Almost sure feasibility means that $Z_i$ lies in the closed convex feasibility cone $K_\alpha(a_i, b_i)$ with probability $1$. By convexity of $K_\alpha(a_i, b_i)$, its expectation $\bar{Z}_i$ also belongs to $K_\alpha(a_i, b_i)$ for every $i$, so $(\bar{Z}_0, \bar{Z}_1, \ldots, \bar{Z}_k)$ is a feasible deterministic trajectory for $X$. By linearity of expectation and Jensen's inequality (convexity of the Euclidean norm $\|\cdot\|_2$),
\[
\sum_{i=1}^k \|\bar{Z}_i - \bar{Z}_{i-1}\|_2
= \sum_{i=1}^k \big\|\mathbb{E}[Z_i - Z_{i-1}]\big\|_2
\le \sum_{i=1}^k \mathbb{E}\big[\|Z_i - Z_{i-1}\|_2\big]
= \mathbb{E}[\mathcal{R}(X)].
\]
Thus every competitive ratio achieved in expectation by $\mathcal{R}$ is bounded below by the deterministic infimum $\rho^*_{\mathrm{det}}(\alpha)$, so $\rho^*_{\mathrm{rand}}(\alpha) \ge \rho^*_{\mathrm{det}}(\alpha)$. The reverse inequality holds because deterministic algorithms are a special case of randomized algorithms, and the lower bound $F_\alpha(s^*)$ follows from Theorem~\ref{thm:lb-universal}.
\end{proof}

At $\alpha = \pi/4$, the quartic~\eqref{eq:quartic-p} factors into elementary radicals and combines with our upper bound for \textsc{Midpoint} to bracket the optimal competitive ratio, both deterministic and randomized, within a narrow interval of width less than $0.018$.

\begin{corollary}[Competitive Ratio Gap at $\alpha = \pi/4$]\label{cor:pi4-sandwich}
At scanning half-angle $\alpha = \pi/4$, the optimal competitive ratio satisfies
\[
\frac{1 + \sqrt{2}}{2} \le \rho^*_{\mathrm{rand}}(\pi/4) = \rho^*_{\mathrm{det}}(\pi/4) \le \sqrt{\frac{3}{2}},
\]
or numerically $1.2071 \le \rho^*(\pi/4) \le 1.2247$.
\end{corollary}

\begin{proof}
At $\alpha = \pi/4$, $\sin(2\alpha) = 1$ and $\cos(4\alpha) = -1$, so~\eqref{eq:quartic-p} factors as
\[
P_{\pi/4}(s) = s^4 - 2s^3 - 2s - 1 = (s^2 + 1)(s^2 - 2s - 1).
\]
Its unique root in $(1, \infty)$ is $s^* = 1 + \sqrt{2}$, and substituting into~\eqref{eq:lb-star} gives
\[
F_{\pi/4}(1 + \sqrt{2}) = \frac{(1 + \sqrt{2})(2 + \sqrt{2})}{\sqrt{1 + (1 + \sqrt{2})^4 + 2(1 + \sqrt{2})^2}}
= \frac{(1 + \sqrt{2})(2 + \sqrt{2})}{1 + (1 + \sqrt{2})^2}
= \frac{1 + \sqrt{2}}{2}.
\]
Combining Theorem~\ref{thm:lb-universal}, Lemma~\ref{lem:derandomization}, and the upper bound for \textsc{Midpoint} (Theorem~\ref{thm:midpoint}) yields the stated bounds; squaring gives $\frac{3 + 2\sqrt{2}}{4} < \frac{3}{2}$ since $2\sqrt{2} < 3$.
\end{proof}

\begin{figure}[t]
\centering
\begin{tikzpicture}[xscale=4.8, yscale=3.8, >=stealth]
  % Colors
  \definecolor{colLB}{RGB}{15, 76, 129}      % Deep navy: universal lower bound
  \definecolor{colAlam}{RGB}{192, 41, 66}    % Crimson: A_{lambda*}
  \definecolor{colBeta}{RGB}{118, 68, 138}   % Purple: beta*-Hedge
  \definecolor{colGreedy}{RGB}{55, 65, 81}   % Charcoal: Greedy
  \definecolor{colSU}{RGB}{13, 148, 136}     % Teal: Straight-Up
  \definecolor{colCA}{RGB}{180, 120, 40}     % Ochre: Cone-Apex

  % Subtle horizontal reference lines (y = rho - 1)
  \draw[gray!22, line width=0.35pt] (0, 0.207107) -- (1.570796, 0.207107);
  \draw[gray!22, line width=0.35pt] (0, 0.224745) -- (1.570796, 0.224745);
  \draw[gray!22, line width=0.35pt] (0, 0.250000) -- (1.570796, 0.250000);
  \draw[gray!22, line width=0.35pt] (0, 0.414214) -- (1.570796, 0.414214);
  \draw[gray!22, line width=0.35pt] (0, 0.732051) -- (1.570796, 0.732051);
  \draw[gray!22, line width=0.35pt] (0, 1.000000) -- (1.570796, 1.000000);

  % Subtle vertical landmark guides
  \draw[gray!25, densely dotted, line width=0.4pt] (0.523599, 0) -- (0.523599, 1.15);
  \draw[gray!25, densely dotted, line width=0.4pt] (0.785398, 0) -- (0.785398, 1.15);
  \draw[gray!25, densely dotted, line width=0.4pt] (1.047198, 0) -- (1.047198, 1.15);

  % Axes
  \draw[->, line width=0.55pt] (0, 0) -- (1.64, 0) node[right, font=\footnotesize] {$\alpha$};
  \draw[->, line width=0.55pt] (0, 0) -- (0, 1.21) node[above, font=\footnotesize] {$\rho$};

  % X-axis ticks
  \foreach \x/\lbl in {
    0/$0$,
    0.523599/$\frac{\pi}{6}$,
    0.785398/$\frac{\pi}{4}$,
    1.047198/$\frac{\pi}{3}$,
    1.570796/$\frac{\pi}{2}$
  } {
    \draw[line width=0.5pt] (\x, 0) -- (\x, -0.022) node[below, font=\scriptsize] {\lbl};
  }

  % Y-axis ticks (y = rho - 1)
  \foreach \y/\lbl in {
    0.000000/$1.0$,
    0.250000/$1.25$,
    0.414214/$\sqrt{2}$,
    0.732051/$\sqrt{3}$,
    1.000000/$2.0$
  } {
    \draw[line width=0.5pt] (0, \y) -- (-0.018, \y) node[left, font=\scriptsize] {\lbl};
  }

  % 1. Cone-Apex: sec(alpha)
  \draw[colCA, line width=0.75pt, densely dotted] plot[smooth] coordinates {
    (0.000000,0.000000) (0.014544,0.000106) (0.029089,0.000423) (0.043633,0.000953)
    (0.058178,0.001695) (0.072722,0.002650) (0.087266,0.003820) (0.101811,0.005205)
    (0.116355,0.006808) (0.130900,0.008629) (0.145444,0.010671) (0.159989,0.012936)
    (0.174533,0.015427) (0.189077,0.018145) (0.203622,0.021095) (0.218166,0.024280)
    (0.232711,0.027702) (0.247255,0.031366) (0.261799,0.035276) (0.276344,0.039437)
    (0.290888,0.043853) (0.305433,0.048529) (0.319977,0.053471) (0.334521,0.058686)
    (0.349066,0.064178) (0.363610,0.069955) (0.378155,0.076024) (0.392699,0.082392)
    (0.407243,0.089068) (0.421788,0.096060) (0.436332,0.103378) (0.450877,0.111030)
    (0.465421,0.119028) (0.479966,0.127382) (0.494510,0.136104) (0.509054,0.145205)
    (0.523599,0.154701) (0.538143,0.164603) (0.552688,0.174927) (0.567232,0.185689)
    (0.581776,0.196906) (0.596321,0.208594) (0.610865,0.220775) (0.625410,0.233466)
    (0.639954,0.246691) (0.654498,0.260472) (0.669043,0.274834) (0.683587,0.289803)
    (0.698132,0.305407) (0.712676,0.321677) (0.727221,0.338643) (0.741765,0.356342)
    (0.756309,0.374809) (0.770854,0.394086) (0.785398,0.414214) (0.799943,0.435239)
    (0.814487,0.457213) (0.829031,0.480187) (0.843576,0.504221) (0.858120,0.529377)
    (0.872665,0.555724) (0.887209,0.583335) (0.901753,0.612291) (0.916298,0.642680)
    (0.930842,0.674597) (0.945387,0.708148) (0.959931,0.743447) (0.974475,0.780620)
    (0.989020,0.819806) (1.003564,0.861159) (1.018109,0.904847) (1.032653,0.951058)
    (1.047198,1.000000) (1.061742,1.051906) (1.076286,1.107036) (1.087030,1.150000)
  };

  % 2. Straight-Up: max(2 cos alpha, csc alpha)
  \draw[colSU, line width=0.8pt, dash pattern=on 3.2pt off 1.5pt on 0.9pt off 1.5pt] plot[smooth] coordinates {
    (0.000000,1.000000) (0.014544,0.999788) (0.029089,0.999154) (0.043633,0.998096)
    (0.058178,0.996616) (0.072722,0.994714) (0.087266,0.992389) (0.101811,0.989643)
    (0.116355,0.986477) (0.130900,0.982890) (0.145444,0.978883) (0.159989,0.974458)
    (0.174533,0.969616) (0.189077,0.964356) (0.203622,0.958681) (0.218166,0.952592)
    (0.232711,0.946090) (0.247255,0.939176) (0.261799,0.931852) (0.276344,0.924119)
    (0.290888,0.915979) (0.305433,0.907434) (0.319977,0.898485) (0.334521,0.889135)
    (0.349066,0.879385) (0.363610,0.869238) (0.378155,0.858695) (0.392699,0.847759)
    (0.407243,0.836432) (0.421788,0.824717) (0.436332,0.812616) (0.450877,0.800131)
    (0.465421,0.787265) (0.479966,0.774022) (0.494510,0.760403) (0.509054,0.746412)
    (0.523599,0.732051) (0.538143,0.717324) (0.552688,0.702233) (0.567232,0.686783)
    (0.581776,0.670976) (0.596321,0.654815) (0.610865,0.638304) (0.625410,0.621447)
    (0.639954,0.604246) (0.654498,0.586707) (0.669043,0.568831) (0.683587,0.550624)
    (0.698132,0.532089) (0.712676,0.513230) (0.727221,0.494050) (0.741765,0.474555)
    (0.756309,0.454747) (0.770854,0.434632) (0.785398,0.414214)
  };
  \draw[colSU, line width=0.8pt, dash pattern=on 3.2pt off 1.5pt on 0.9pt off 1.5pt] plot[smooth] coordinates {
    (0.785398,0.414214) (0.799943,0.394086) (0.814487,0.374809) (0.829031,0.356342)
    (0.843576,0.338643) (0.858120,0.321677) (0.872665,0.305407) (0.887209,0.289803)
    (0.901753,0.274834) (0.916298,0.260472) (0.930842,0.246691) (0.945387,0.233466)
    (0.959931,0.220775) (0.974475,0.208594) (0.989020,0.196906) (1.003564,0.185689)
    (1.018109,0.174927) (1.032653,0.164603) (1.047198,0.154701) (1.061742,0.145205)
    (1.076286,0.136104) (1.090831,0.127382) (1.105375,0.119028) (1.119920,0.111030)
    (1.134464,0.103378) (1.149008,0.096060) (1.163553,0.089068) (1.178097,0.082392)
    (1.192642,0.076024) (1.207186,0.069955) (1.221730,0.064178) (1.236275,0.058686)
    (1.250819,0.053471) (1.265364,0.048529) (1.279908,0.043853) (1.294453,0.039437)
    (1.308997,0.035276) (1.323541,0.031366) (1.338086,0.027702) (1.352630,0.024280)
    (1.367175,0.021095) (1.381719,0.018145) (1.396263,0.015427) (1.410808,0.012936)
    (1.425352,0.010671) (1.439897,0.008629) (1.454441,0.006808) (1.468985,0.005205)
    (1.483530,0.003820) (1.498074,0.002650) (1.512619,0.001695) (1.527163,0.000953)
    (1.541708,0.000423) (1.556252,0.000106) (1.570796,0.000000)
  };

  % 3. Greedy: sec(alpha) on (0, pi/4], sqrt(2 - cos^2(2*alpha)) on (pi/4, pi/2)
  \draw[colGreedy, line width=0.8pt, dash pattern=on 3pt off 1.8pt] plot[smooth] coordinates {
    (0.000000,0.000000) (0.043633,0.000953) (0.087266,0.003820) (0.130900,0.008629)
    (0.174533,0.015427) (0.218166,0.024280) (0.261799,0.035276) (0.305433,0.048529)
    (0.349066,0.064178) (0.392699,0.082392) (0.436332,0.103378) (0.479966,0.127382)
    (0.523599,0.154701) (0.567232,0.185689) (0.610865,0.220775) (0.654498,0.260472)
    (0.698132,0.305407) (0.741765,0.356342) (0.785398,0.414214)
  };
  \draw[colGreedy, line width=0.8pt, dash pattern=on 3pt off 1.8pt] plot[smooth] coordinates {
    (0.785398,0.414214) (0.872665,0.403512) (0.959931,0.372233) (1.047198,0.322876)
    (1.134464,0.259692) (1.221730,0.188771) (1.308997,0.118034) (1.396263,0.056872)
    (1.483530,0.014965) (1.570796,0.000000)
  };

  % 4. Optimal fixed-angle beta*-Hedge
  \draw[colBeta, line width=0.85pt, densely dashed] plot[smooth] coordinates {
    (0.000000,0.000000) (0.014544,0.000106) (0.029089,0.000423) (0.043633,0.000953)
    (0.058178,0.001695) (0.072722,0.002650) (0.087266,0.003820) (0.101811,0.005205)
    (0.116355,0.006808) (0.130900,0.008629) (0.145444,0.010671) (0.159989,0.012936)
    (0.174533,0.015427) (0.189077,0.018145) (0.203622,0.021095) (0.218166,0.024280)
    (0.232711,0.027702) (0.247255,0.031366) (0.261799,0.035276) (0.276344,0.039437)
    (0.290888,0.043853) (0.305433,0.048529) (0.319977,0.053471) (0.334521,0.058686)
    (0.349066,0.064178) (0.363610,0.069955) (0.378155,0.076024) (0.392699,0.082392)
    (0.407243,0.089068) (0.421788,0.096060) (0.436332,0.103378) (0.450877,0.111030)
    (0.465421,0.119028) (0.479966,0.127382) (0.494510,0.136104) (0.509054,0.145205)
    (0.523599,0.154701) (0.538143,0.164227) (0.552688,0.173390) (0.567232,0.182152)
    (0.581776,0.190483) (0.596321,0.198353) (0.610865,0.205737) (0.625410,0.212610)
    (0.639954,0.218953) (0.654498,0.224745) (0.669043,0.229970) (0.683587,0.234614)
    (0.698132,0.238664) (0.712676,0.242109) (0.727221,0.244940) (0.741765,0.247150)
    (0.756309,0.248732) (0.770854,0.249683) (0.785398,0.250000) (0.799943,0.249683)
    (0.814487,0.248732) (0.829031,0.247150) (0.843576,0.244940) (0.858120,0.242109)
    (0.872665,0.238664) (0.887209,0.234614) (0.901753,0.229970) (0.916298,0.224745)
    (0.930842,0.218953) (0.945387,0.212610) (0.959931,0.205737) (0.974475,0.198353)
    (0.989020,0.190483) (1.003564,0.182152) (1.018109,0.173390) (1.032653,0.164227)
    (1.047198,0.154701) (1.061742,0.145205) (1.076286,0.136104) (1.090831,0.127382)
    (1.105375,0.119028) (1.119920,0.111030) (1.134464,0.103378) (1.149008,0.096060)
    (1.163553,0.089068) (1.178097,0.082392) (1.192642,0.076024) (1.207186,0.069955)
    (1.221730,0.064178) (1.236275,0.058686) (1.250819,0.053471) (1.265364,0.048529)
    (1.279908,0.043853) (1.294453,0.039437) (1.308997,0.035276) (1.323541,0.031366)
    (1.338086,0.027702) (1.352630,0.024280) (1.367175,0.021095) (1.381719,0.018145)
    (1.396263,0.015427) (1.410808,0.012936) (1.425352,0.010671) (1.439897,0.008629)
    (1.454441,0.006808) (1.468985,0.005205) (1.483530,0.003820) (1.498074,0.002650)
    (1.512619,0.001695) (1.527163,0.000953) (1.541708,0.000423) (1.556252,0.000106)
    (1.570796,0.000000)
  };

  % 5. State-based interpolation A_{lambda*}
  \draw[colAlam, line width=0.95pt] plot[smooth] coordinates {
    (0.000000,0.000000) (0.014544,0.000106) (0.029089,0.000423) (0.043633,0.000953)
    (0.058178,0.001695) (0.072722,0.002650) (0.087266,0.003820) (0.101811,0.005205)
    (0.116355,0.006806) (0.130900,0.008626) (0.145444,0.010666) (0.159989,0.012927)
    (0.174533,0.015411) (0.189077,0.018120) (0.203622,0.021056) (0.218166,0.024219)
    (0.232711,0.027611) (0.247255,0.031234) (0.261799,0.035087) (0.276344,0.039170)
    (0.290888,0.043484) (0.305433,0.048025) (0.319977,0.052793) (0.334521,0.057782)
    (0.349066,0.062989) (0.363610,0.068406) (0.378155,0.074026) (0.392699,0.079838)
    (0.407243,0.085831) (0.421788,0.091989) (0.436332,0.098296) (0.450877,0.104735)
    (0.465421,0.111283) (0.479966,0.117917) (0.494510,0.124612) (0.509054,0.131341)
    (0.523599,0.138075) (0.538143,0.144783) (0.552688,0.151435) (0.567232,0.157997)
    (0.581776,0.164437) (0.596321,0.170723) (0.610865,0.176822) (0.625410,0.182702)
    (0.639954,0.188330) (0.654498,0.193678) (0.669043,0.198715) (0.683587,0.203414)
    (0.698132,0.207747) (0.712676,0.211689) (0.727221,0.215216) (0.741765,0.218305)
    (0.756309,0.220936) (0.770854,0.223089) (0.785398,0.224745) (0.799943,0.225887)
    (0.814487,0.226499) (0.829031,0.226567) (0.843576,0.226076) (0.858120,0.225014)
    (0.872665,0.223369) (0.887209,0.221130) (0.901753,0.218287) (0.916298,0.214829)
    (0.930842,0.210749) (0.945387,0.206036) (0.959931,0.200684) (0.974475,0.194683)
    (0.989020,0.188027) (1.003564,0.180708) (1.018109,0.172719) (1.032653,0.164052)
    (1.047198,0.154701) (1.061742,0.145205) (1.076286,0.136104) (1.090831,0.127382)
    (1.105375,0.119028) (1.119920,0.111030) (1.134464,0.103378) (1.149008,0.096060)
    (1.163553,0.089068) (1.178097,0.082392) (1.192642,0.076024) (1.207186,0.069955)
    (1.221730,0.064178) (1.236275,0.058686) (1.250819,0.053471) (1.265364,0.048529)
    (1.279908,0.043853) (1.294453,0.039437) (1.308997,0.035276) (1.323541,0.031366)
    (1.338086,0.027702) (1.352630,0.024280) (1.367175,0.021095) (1.381719,0.018145)
    (1.396263,0.015427) (1.410808,0.012936) (1.425352,0.010671) (1.439897,0.008629)
    (1.454441,0.006808) (1.468985,0.005205) (1.483530,0.003820) (1.498074,0.002650)
    (1.512619,0.001695) (1.527163,0.000953) (1.541708,0.000423) (1.556252,0.000106)
    (1.570796,0.000000)
  };

  % 6. Universal Lower Bound F_alpha(s*)
  \draw[colLB, line width=0.9pt] plot[smooth] coordinates {
    (0.000000,0.000000) (0.014544,0.000106) (0.029089,0.000423) (0.043633,0.000953)
    (0.058178,0.001695) (0.072722,0.002650) (0.087266,0.003820) (0.101811,0.005205)
    (0.116355,0.006806) (0.130900,0.008626) (0.145444,0.010666) (0.159989,0.012927)
    (0.174533,0.015411) (0.189077,0.018120) (0.203622,0.021056) (0.218166,0.024219)
    (0.232711,0.027611) (0.247255,0.031233) (0.261799,0.035086) (0.276344,0.039169)
    (0.290888,0.043481) (0.305433,0.048020) (0.319977,0.052785) (0.334521,0.057770)
    (0.349066,0.062970) (0.363610,0.068377) (0.378155,0.073982) (0.392699,0.079773)
    (0.407243,0.085735) (0.421788,0.091853) (0.436332,0.098105) (0.450877,0.104469)
    (0.465421,0.110919) (0.479966,0.117426) (0.494510,0.123960) (0.509054,0.130486)
    (0.523599,0.136970) (0.538143,0.143373) (0.552688,0.149658) (0.567232,0.155787)
    (0.581776,0.161722) (0.596321,0.167424) (0.610865,0.172859) (0.625410,0.177991)
    (0.639954,0.182789) (0.654498,0.187222) (0.669043,0.191263) (0.683587,0.194887)
    (0.698132,0.198072) (0.712676,0.200799) (0.727221,0.203052) (0.741765,0.204818)
    (0.756309,0.206087) (0.770854,0.206852) (0.785398,0.207107) (0.799943,0.206852)
    (0.814487,0.206087) (0.829031,0.204818) (0.843576,0.203052) (0.858120,0.200799)
    (0.872665,0.198072) (0.887209,0.194887) (0.901753,0.191263) (0.916298,0.187222)
    (0.930842,0.182789) (0.945387,0.177991) (0.959931,0.172859) (0.974475,0.167424)
    (0.989020,0.161722) (1.003564,0.155787) (1.018109,0.149658) (1.032653,0.143373)
    (1.047198,0.136970) (1.061742,0.130486) (1.076286,0.123960) (1.090831,0.117426)
    (1.105375,0.110919) (1.119920,0.104469) (1.134464,0.098105) (1.149008,0.091853)
    (1.163553,0.085735) (1.178097,0.079773) (1.192642,0.073982) (1.207186,0.068377)
    (1.221730,0.062970) (1.236275,0.057770) (1.250819,0.052785) (1.265364,0.048020)
    (1.279908,0.043481) (1.294453,0.039169) (1.308997,0.035086) (1.323541,0.031233)
    (1.338086,0.027611) (1.352630,0.024219) (1.367175,0.021056) (1.381719,0.018120)
    (1.396263,0.015411) (1.410808,0.012927) (1.425352,0.010666) (1.439897,0.008626)
    (1.454441,0.006806) (1.468985,0.005205) (1.483530,0.003820) (1.498074,0.002650)
    (1.512619,0.001695) (1.527163,0.000953) (1.541708,0.000423) (1.556252,0.000106)
    (1.570796,0.000000)
  };

  % Landmark points (crisp inner sep=0.85pt)
  \node[circle, fill=colSU,     inner sep=0.85pt] at (0.523599, 0.732051) {};
  \node[circle, fill=colSU,     inner sep=0.85pt] at (0.785398, 0.414214) {};
  \node[circle, fill=colCA,     inner sep=0.85pt] at (1.047198, 1.000000) {};
  \node[circle, fill=colGreedy, inner sep=0.85pt] at (1.047198, 0.322876) {};
  \node[circle, fill=colBeta,   inner sep=0.85pt] at (0.785398, 0.250000) {};
  \node[circle, fill=colAlam,   inner sep=0.85pt] at (0.785398, 0.224745) {};
  \node[circle, fill=colLB,     inner sep=0.85pt] at (0.785398, 0.207107) {};
  \node[circle, fill=colAlam,   inner sep=0.85pt] at (1.047198, 0.154701) {};
  \node[circle, fill=colLB,     inner sep=0.85pt] at (1.047198, 0.136970) {};

  % Landmark annotations
  \node[font=\scriptsize, colSU,     above right] at (0.530, 0.732) {$\sqrt{3}$};
  \node[font=\scriptsize, colSU,     above right] at (0.792, 0.415) {$\sqrt{2}$};
  \node[font=\scriptsize, colCA,     left]        at (1.038, 1.000) {$2$};
  \node[font=\scriptsize, colGreedy, above right] at (1.054, 0.322) {$\frac{\sqrt{7}}{2}$};
  \node[font=\scriptsize, colBeta,   above]       at (0.785, 0.258) {$\frac{5}{4}$};
  \node[font=\scriptsize, colAlam,   above right] at (0.818, 0.228) {$\sqrt{\frac{3}{2}}$};
  \node[font=\scriptsize, colLB,     below]       at (0.785, 0.195) {$\frac{1+\sqrt{2}}{2}$};
  \node[font=\scriptsize, colAlam,   above right] at (1.055, 0.155) {$\frac{2}{\sqrt{3}}$};
  \node[font=\scriptsize, colLB,     below left]  at (1.042, 0.134) {$\approx 1.137$};

  % Legend (upper right, clear of all curves)
  \begin{scope}[shift={(1.13, 0.50)}]
    \draw[fill=white, draw=gray!35, line width=0.4pt] (0, 0) rectangle (0.44, 0.62);
    \draw[colCA,     line width=0.75pt, densely dotted] (0.025, 0.55) -- (0.115, 0.55);
    \node[right, font=\scriptsize] at (0.125, 0.55) {\textsc{Cone-Apex}};
    \draw[colSU,     line width=0.8pt, dash pattern=on 3.2pt off 1.5pt on 0.9pt off 1.5pt] (0.025, 0.45) -- (0.115, 0.45);
    \node[right, font=\scriptsize] at (0.125, 0.45) {\textsc{Straight-Up}};
    \draw[colGreedy, line width=0.8pt, dash pattern=on 3pt off 1.8pt] (0.025, 0.35) -- (0.115, 0.35);
    \node[right, font=\scriptsize] at (0.125, 0.35) {\textsc{Greedy}};
    \draw[colBeta,   line width=0.85pt, densely dashed] (0.025, 0.25) -- (0.115, 0.25);
    \node[right, font=\scriptsize] at (0.125, 0.25) {$\beta^*$-\textsc{Hedge}};
    \draw[colAlam,   line width=0.95pt] (0.025, 0.15) -- (0.115, 0.15);
    \node[right, font=\scriptsize] at (0.125, 0.15) {$\mathcal{A}_{\lambda^*}$};
    \draw[colLB,     line width=0.9pt]  (0.025, 0.05) -- (0.115, 0.05);
    \node[right, font=\scriptsize] at (0.125, 0.05) {Lower bound $F_\alpha(s^*)$};
  \end{scope}
\end{tikzpicture}
\caption{Competitive ratios across $\alpha \in (0, \pi/2)$: baseline algorithms (\textsc{Cone-Apex}, \textsc{Straight-Up}, and \textsc{Greedy}), optimal fixed-angle $\beta^*$-\textsc{Hedge}, optimal state-based interpolation $\mathcal{A}_{\lambda^*}$, and the universal deterministic lower bound $F_\alpha(s^*)$.}
\label{fig:cr-comparison}
\end{figure}

\section{Conclusions and Open Problems}\label{sec:conclusions}

We have studied the online target coverage problem for an aerial drone equipped with a downward-facing sensor of fixed half-angle $\alpha \in (0, \pi/2)$ monitoring targets arriving sequentially on a line. By reformulating physical drone positions $(x, y)$ in terms of left- and right-reach coverage coordinates $(u, v) = (y\tan\alpha - x,\, y\tan\alpha + x)$ equipped with the induced norm $N_\alpha(u, v)$, we obtained a unified geometric framework that governs both offline optimal placement and online competitive analysis across the entire open range $\alpha \in (0, \pi/2)$.

Within this framework, we established tight competitive ratios for the memoryless baseline strategies \textsc{Straight-Up}, \textsc{Cone-Apex}
and \textsc{Greedy}. Next, we analyzed fixed-angle steering ($\beta\text{-}\textsc{Hedge}$), determined its uniquely optimal steering angle $\beta^*(\alpha)$, and proved that state-based linear interpolation ($\mathcal{A}_\lambda$, including \textsc{Midpoint} at $\alpha = \pi/4$) strictly outperforms $\beta^*\text{-}\textsc{Hedge}$ for all $\alpha \in (0, \pi/3)$ by shedding excess altitude during non-dominant updates. Finally, using a scale-invariant alternating request sequence together with an averaging and convexity-based derandomization argument, we derived a universal lower bound valid for all deterministic and randomized online algorithms against an oblivious adversary across $\alpha \in (0, \pi/2)$, attaining its maximum value of $\frac{1 + \sqrt{2}}{2} \approx 1.2071$ at $\alpha = \pi/4$. Figure~\ref{fig:cr-comparison} illustrates all the competitive ratios derived in this paper as well as the universal lower bound.

Our results suggest several natural directions for future investigation:

\begin{enumerate}
    \item \textbf{Closing the Competitive Gap on the Line.}
    A modest gap persists between the optimal state-based linear interpolation bound $\min_{\lambda \in [0,1]} \rho(\mathcal{A}_\lambda; \alpha)$ and our universal lower bound across $\alpha \in (0, \pi/2)$, most notably between $\frac{1 + \sqrt{2}}{2} \approx 1.2071$ and $\sqrt{3/2} \approx 1.2247$ at $\alpha = \pi/4$. Because $\mathcal{A}_\lambda$ restricts the drone to a piecewise-affine mapping of the spanning interval $[-a, b]$, it remains open whether nonlinear state-dependent positioning rules, potential-function algorithms, or Work-Function strategies can close this gap.

    \item \textbf{Two-Dimensional Ground Targets and Three-Dimensional Flight.}
    A natural geometric generalization places target requests in the two-dimensional ground plane $\mathbb{R}^2$, with the drone navigating the three-dimensional upper half-space $\mathbb{R}^2 \times \mathbb{R}_{\ge 0}$. Unlike the one-dimensional setting, where the minimal spanning set of all requests seen so far is completely determined by two extreme endpoints $(-a_i, b_i)$, covering a planar point set $X_i \subset \mathbb{R}^2$ from $(x_i, y_i, z_i)$ requires enclosing its entire convex hull $\mathrm{conv}(X_i)$ within a disk of radius $z_i\tan\alpha$ for a conical sensor. Feasibility at round $i$ therefore constrains the altitude via $z_i \ge \cot\alpha \max_{X \in X_i} \|(x_i, y_i) - X\|_2$, turning the problem into chasing an intersection of three-dimensional cones under Euclidean path length. Determining whether constant-competitive online algorithms exist in $\mathbb{R}^3$ is a compelling open problem closely related to online minimum enclosing ball.

    \item \textbf{Altitude-Weighted Costs and Climb Penalties.}
    Our model charges uniform Euclidean distance regardless of flight direction or elevation and permits arbitrarily high altitude. In physical deployments, ascending against gravity consumes substantially more energy than horizontal or downward flight, high-altitude operation degrades sensor resolution, and airspace regulations impose a maximum ceiling $y \le H_{\max}$. One can model these effects by replacing isotropic Euclidean distance with an anisotropic step cost $\sqrt{dx^2 + \kappa\, dy^2}$ ($\kappa > 1$) or a height-weighted path integral $\int_\gamma w(y)\sqrt{dx^2 + dy^2}$ for a non-decreasing penalty $w(y)$. Crucially, nonlinear altitude weights destroy the scale invariance underlying both $\mathcal{A}_\lambda$ and our geometric lower-bound construction, while a hard ceiling $y \le H_{\max}$ renders wide-span instances infeasible for simultaneous static coverage and forces a transition to dynamic patrolling.

    \item \textbf{Partial ($p$-Fraction) Online Coverage and Problem Well-Definedness.}
    Suppose the drone is only required to cover a prescribed fraction $p \in (0, 1)$ (e.g., $p = 0.8$) of the targets rather than all of them.
    Partial online coverage becomes mathematically well-posed only under refined models such as (i) \emph{bicriteria (resource-augmented)} comparisons, where an online algorithm maintaining $p$-coverage at every round is benchmarked against an offline optimum required to cover a stricter fraction $p' > p$ (e.g., $80\%$ online versus $100\%$ offline), (ii) \emph{outlier-budget or prize-collecting} formulations with a fixed rejection quota $k$ or per-target penalty, or (iii) \emph{measure-based} coverage, where the drone must cover at least a $p$-fraction of the total length or area of the revealed target set.

\end{enumerate}

\noindent
\textbf{Acknowledgment:} The revision used AI models for assistance in proof development, adversarial review of results, and manuscript preparation.

\bibliographystyle{abbrvnat} % Changed from plainurl
\bibliography{refs}

\end{document}